\documentclass[11pt]{article}
\usepackage{times}
\usepackage{latexsym}
\usepackage{epsfig,enumerate,amsmath,amsfonts,amssymb,setspace}
\usepackage{indentfirst}
\usepackage{wrapfig,lipsum}
\usepackage{epsfig}
\usepackage{graphicx}
\usepackage{amsthm}
\usepackage{algorithm}
\usepackage{algorithmic}
\usepackage[T1]{fontenc}
\usepackage{authblk}
\usepackage[algo2e,ruled,vlined]{algorithm2e}
\newtheorem{theorem}{Theorem}[section]
\newtheorem{lemma}{Lemma}[section]
\newtheorem{corollary}{Corollary}[section]

\newtheorem{Ob}{Observation}[section]

\title{Algorithms for Maximum Internal Spanning Tree Problem for Some Graph Classes \footnote{The third author is supported by an NSF Graduate Research Fellowship under Grant No. DGE-1650044.}}

\author[1]{Gopika Sharma \thanks{2017maz0007@iitrpr.ac.in}}
\author[1]{Arti Pandey\thanks{arti@iitrpr.ac.in}}
\author[2]{Michael C. Wigal \thanks{wigal@gatech.edu}}
\affil[1]{Department of Mathematics, Indian Institute of Technology Ropar, Punjab, India.}
\affil[2]{School of Mathematics, Georgia Institute of Technology, Atlanta, GA, USA.}

\date{}
\begin{document}
\maketitle
\begin{abstract}
For a given graph $G$, a maximum internal spanning tree of $G$ is a spanning tree of $G$ with maximum number of internal vertices. The \textsc{Maximum Internal Spanning Tree (MIST)} problem is to find a maximum internal spanning tree of the given graph. The MIST problem is a generalization of the Hamiltonian path problem. Since the Hamiltonian path problem is NP-hard, even for bipartite and chordal graphs, two important subclasses of graphs, the MIST problem also remains NP-hard for these graph classes. In this paper, we propose linear-time algorithms to compute a maximum internal spanning tree of cographs, block graphs, cactus graphs, chain graphs and bipartite permutation graphs. The optimal path cover problem, which asks to find a path cover of the given graph with maximum number of edges, is also a well studied problem. In this paper, we also study the relationship between the number of internal vertices in maximum internal spanning tree and number of edges in optimal path cover for the special graph classes mentioned above.
\end{abstract} 

\section{Introduction}\label{sec:intro}
The Maximum Internal Spanning Tree (MIST) problem is a degree based spanning tree optimization problem, in which we ask to find a spanning tree of a given graph such that the number of vertices of degree at least two is maximized. The MIST problem is motivated by telecommunication network design \cite{salamon}.  We also believe that MIST
problem has its own theoretical importance as it is a generalization of the Hamiltonian path problem, a known NP-complete problem \cite{garey}. The Hamiltonian path problem remains NP-complete for chordal graphs and chordal bipartite graphs \cite{laii, mull}. Hence, we also do not expect polynomial time algorithms for the MIST problem in chordal graphs and chordal bipartite graphs.

The dual problem to MIST, the \textsc{Minimum Leaves Spanning Tree (MLST)} problem asks to find a spanning tree with minimum number of leaves for a given graph. The MLST problem cannot be approximated  within any constant factor unless P=NP~\cite{lu1}. Unlike MLST, several constant factor approximation algorithms have been proposed for the MIST problem in literature. In $2003$, Prieto et al. \cite{prieto} gave a $2$-approximation algorithm for the MIST problem whose running time was later improved by Salamon et al. in 2008 \cite{salam}. Salamon also gave approximation algorithms for claw-free and cubic graphs with approximation factors $\frac{3}{2}$ and $\frac{6}{5}$ respectively \cite{salam}. In 2009, Salamon \cite{sala} gave a $\frac{7}{4}$ -approximation algorithm for graphs with no pendant vertices and later, in 2015, Knauer et al. \cite{knau} showed that a simplified and faster version of Salamon's algorithm yields a $\frac{5}{3}$ -approximation algorithm even on general graphs. In 2014, Li et al. proposed a $\frac{3}{2}$-approximation algorithm using a different approach for general undirected graphs and improved this ratio to $\frac{4}{3}$ for graphs without leaves \cite{li}. Li et al. gave a $\frac{3}{2}$ -approximation algorithm for general graphs using depth-5 local search \cite{li2017}. In 2018, Chen et al. presented a $\frac{17}{13}$ -approximation algorithm which is the best approximation factor till now \cite{chen}. Several FPT-algorithms have also been designed for the MIST problem where the considered parameter is the solution size \cite{prieto, li2017, cohen, fomin, bink, li20152}.

For finding efficient algorithms for the MIST problem, it is often useful to reduce the MIST problem to the path cover problem. A \textit{path cover} $P$ of a graph is a spanning subgraph such that every component of $P$ is a path. A path cover with maximum number of edges is called an \textit{optimal path cover} of $G$. If $P^*$ denotes an optimal path cover of a graph, then number of edges in $P^*$ is denoted by $\vert E(P^*)\vert$. In 2018, Li et al. proposed a polynomial time algorithm for the MIST problem in interval graphs \cite{li2018}. They also proved that number of internal vertices in any MIST of any graph $G$ is at most $\vert E(P^*)\vert - 1$, where $P^*$ is an optimal path cover of $G$. We will observe that number of internal vertices in any MIST of a chain graph is either $\vert E(P^*)\vert - 1$ or $\vert E(P^*)\vert - 2$  and is $\vert E(P^*)\vert - 1$ for cographs. For bipartite permutation, block and cactus graphs, we prove that there is no constant $k$ such that $\vert E(P^*)\vert - k$ is the lower bound on the number of internal vertices in any MIST of such graphs. We also propose linear-time algorithms for the MIST problem in bipartite permutation graphs, block graphs, cactus graphs and cographs. A hierarchy relationship between these classes of graphs is shown in Fig.~\ref{fig:hierarchy}.

The structure of the paper is as follows. In Section \ref{sec:prelim}, we give some basic definitions and notations used in the paper. In Section \ref{sec:blockcactus}, we discuss MIST problem for block graphs and cactus graphs and provide linear-time algorithms for both these graph classes. In Section \ref{sec:cograph}, we prove that MIST of cographs can be computed in linear-time by providing an algorithm. In Section \ref{sec:BP}, we present a linear-time algorithm to find a MIST of an arbitrary bipartite permutation graph. In Section \ref{sec:chain}, we prove a bound for chain graphs regarding number of internal vertices in its MIST.  Finally, Section \ref{sec:conclusion} concludes the paper.
\begin{figure}[h!]
 \centering
\includegraphics[width=8cm, height=3.5cm]{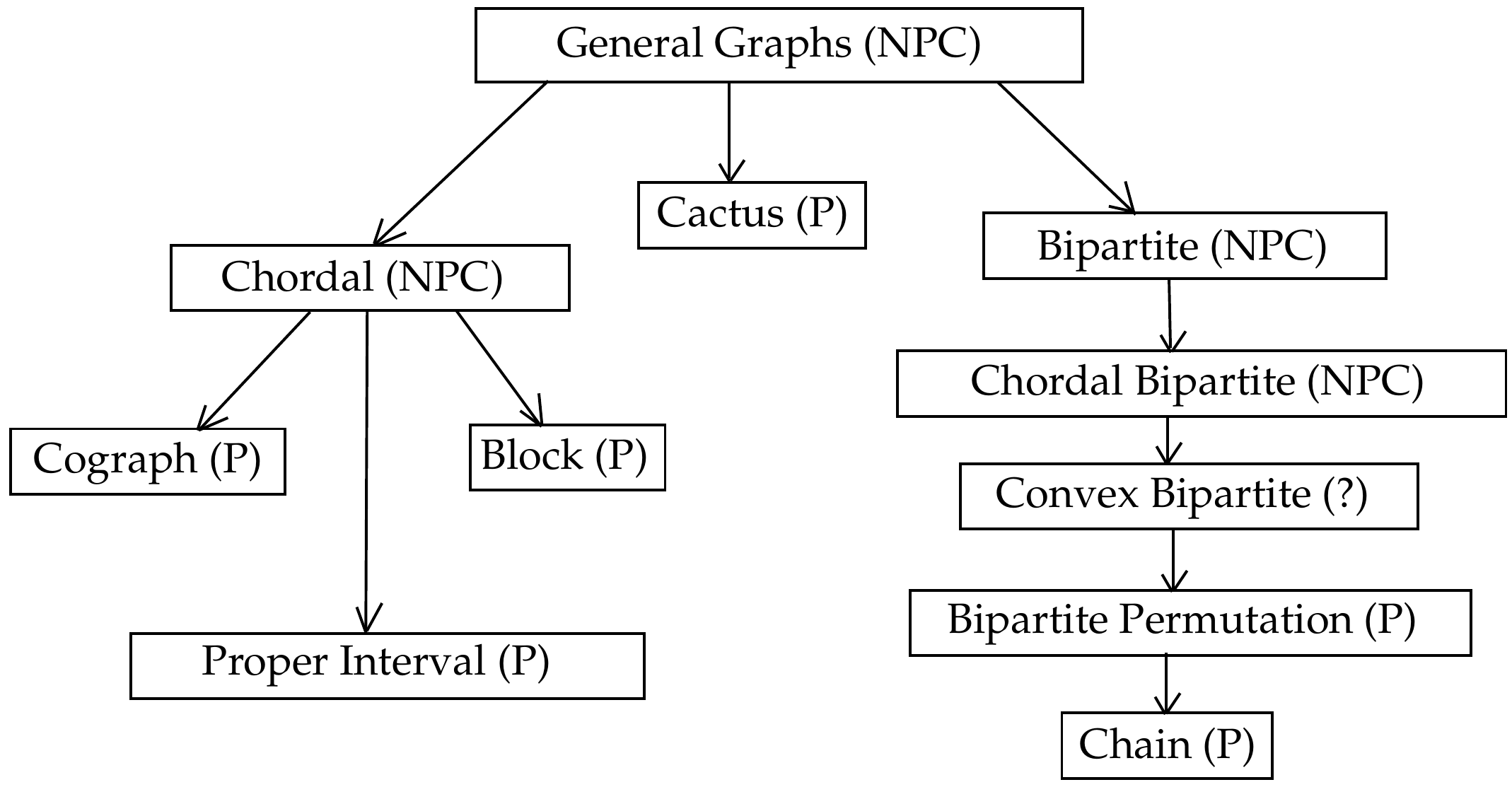}
 \caption{Hierarchy relationship between some classes of graphs}
 \label{fig:hierarchy}
\end{figure}
\section{Preliminaries}
\label{sec:prelim}
Let $G = (V,E)$ be a graph. In this paper, we only consider simple, undirected and connected graphs. For a vertex $u \in G,  d_{G}(u)$ denotes the degree of $u$ in $G$ and $N_G(u)$ denotes the neighborhood of $u$ in $G$. When there is no ambiguity regarding the graph $G$, we simply use $d(u)$ and $N(u)$, to represent the degree of $u$ and neighborhood of $u$, respectively. A vertex $u$ in $V$ is called \textit{pendant} if $d(u) = 1$. The set of pendant vertices in $G$ is denoted by $p(G)$.
The vertex adjacent to a pendant vertex $u$ is called the \textit{support vertex} of $u$, and is denoted by $S(u)$. A vertex $u \in V(G)$ is called \textit{internal}, if $u$ is not pendant, that is,  $d(u) \geq 2$. Let $I(G)$ denotes the set of internal vertices in $G$, and  $i(G) = \vert I(G)\vert$.  For a set $A \subseteq V$ and a spanning tree $T$ of $G$, we define $i_T(A) = \vert{I(T) \cap A}\vert$.

For vertices $x$ and $y$, we denote an edge between $x$ and $y$ by $xy$.  For a subset $S$ of $V(G)$, $G - S$ denotes the subgraph of $G$ obtained after removing vertices of $S$ and edges incident on these vertices from $G$. If $S = \{v\}$, then we simply write $G - v$ for $G - S$. A vertex $v$ of a graph $G$ is called a \textit{cut vertex} if $G - v$  is disconnected.

Throughout this paper, $n$ denotes the number of vertices and $m$ denotes the number of edges in  $G$. A graph $G$ is said to be \textit{bipartite} if $V$ can be partitioned into two disjoint sets $X$ and $Y$ such that every edge of $G$ joins a vertex in $X$ to a vertex in $Y$. Such a partition $(X,Y)$ of $V$ is called a \textit{bipartition}. A bipartite graph with bipartition $(X,Y)$ of $V$ is denoted by $G = (X,Y,E)$. For a set $S\subseteq V$, an \textit{induced subgraph} is the graph whose vertex set is $S$ and edge set consists of all the edges in $E$ that have both the endpoints in $S$, and is denoted by $G[S]$. If $G[C], C \subseteq V$, is a complete subgraph of $G$, then $C$ is called a \textit{clique} of $G$.

A subgraph of $G$ is called a \textit{spanning subgraph} if it contains all the vertices of $G$.  A \textit{path cover} $P$ of a graph is a spanning subgraph such that every component of $P$ is a path. A path cover is an \textit{optimal path cover} if it has the maximum number of edges. A spanning subgraph of $G$ which is also a tree is called a \textit{spanning tree} of $G$. A spanning tree is called a \textit{maximum internal spanning tree(MIST)} if it contains the maximum number of internal vertices among all the spanning trees of $G$. For a graph $G$, the number of internal vertices in any MIST of $G$ is denoted by $Opt(G)$.

Now we state a useful theorem which gives an upper bound on the number of internal vertices in a MIST with respect to the graph's optimal path cover.

\begin{theorem}\label{thm:general upper bound}\cite{li2018} For a graph $G$, the number of internal vertices in a maximum internal spanning tree of $G$  is less than the number of edges in an optimal path cover of $G$, that is,
$Opt(G)\leq \vert E(P^*)\vert - 1$, where $P^*$ denotes an optimal path cover of $G$. Moreover, this bound is tight.
\end{theorem}

Note that the vertices which are pendant in $G$ itself, will be pendant in any MIST of $G$. Hence, we have the following observation.

\begin{Ob}
For a graph $G$, if $v$ is a pendant vertex and $u$ is the support vertex of $v$ in $G$, then $v$ remains a pendant vertex and $u$ remains adjacent support vertex of $v$ in any MIST of $G$.
\end{Ob}

Suppose $G$ is not a tree and $u \in V(G)$ is adjacent to $k$ pendant vertices, say $a_1, \ldots, a_k$. Let $G' = G - \{a_2,\ldots,a_k\}$. Then based on above observation, the number of internal vertices in a MIST of $G$ will be same as the number of internal vertices in any MIST of $G'$. It is also easy to obtain a MIST of $G$ from any MIST of $G'$. Hence, throughout this work, we assume that every vertex has at most one pendant vertex adjacent with it.

Below, we give another result regarding the number of pendant vertices in a spanning tree of a bipartite graph. Note that, if we have $\alpha$ number of internal vertices in a spanning tree of $G$ from one partite set, then at least $\alpha + 1$ vertices must be present in the neighborhood of these $\alpha$ vertices, which lie in the other partite set of the bipartite graph $G$.
\begin{Ob}\label{Ob:basic}
Let $G=(X,Y,E)$ be a bipartite graph with $A \subseteq X \text{ and } B \subseteq Y$. If $N(A) = B$, then there are at least $max\{0, \vert A\vert - \vert B\vert + 1 \}$ pendant vertices from $A$ in any spanning tree of $G$. Similarly, if $N(B) = A$, then there are at least $max\{0, \vert B\vert - \vert A\vert + 1 \}$ pendant vertices from $B$ in any spanning tree of $G$.
\end{Ob}
\section{Block and Cactus Graphs}
\label{sec:blockcactus}
In this section, we discuss the MIST problem for block graphs and cactus graphs. We will show that the MIST problem can be solved in linear-time for both classes of graphs. Block and cactus graphs will also provide our first family of examples in which $Opt(G)$ cannot be lower bounded in terms of $\vert E(P^*)\vert - k$ where $P^*$ is an optimal path cover of $G$ and $k$ is some constant.

A \textit{block} of a graph $G$ is a maximal connected subgraph with no cut vertices. Note that a block of $G$ is either an edge or a 2-connected subgraph. The set of blocks of a graph is called the \textit{block decomposition} of $G$ and is denoted by $B(G)$. Let $B_0 \in B(G)$ and $u,v$ be two vertices belonging to $B_0$, then a path between $u$ and $v$, which contains all the vertices of the block $B_0$, is called a \textit{spanning path between $u$ and $v$ in $B_0$}. We say a block $B$ is \textit{good} if there exists distinct $u,v \in V(B)$ such that both $u$ and $v$ are cut vertices of $G$ and $B$ has a spanning path between $u$ and $v$.  A block is said to be \textit{bad} otherwise. Let $\text{Bad}(G)$ denote the set of bad blocks of $G$.

A \textit{block graph} is a graph in which every block is a clique. If a block graph $G$ contains only one block then $G$ is a complete graph. A block graph is said to be \textit{nontrivial} if it contains at least two blocks. Note that a trivial block has a Hamiltonian path. Thus for the remainder of the section we only consider nontrivial block graphs.

Let $G$ be a nontrivial block graph. Bad blocks of $G$ have another characterization which we state as the Observation \ref{Ob:block_bad}.
\begin{Ob}\label{Ob:block_bad}
A block $B$ of a block graph $G$ is bad if and only if it contains exactly one cut vertex of $G$.
\end{Ob}

A graph $G$ is a \textit{cactus graph} if every block of $G$ is either a cycle or an edge. If a cactus graph $G$ contains only one block then $G$ is either a cycle or an edge and in that case finding a MIST of $G$ is trivial. Again, a cactus graph is said to be \textit{nontrivial} if it contains at least two blocks and now we only consider nontrivial cactus graphs.

Let $G$ be a nontrivial cactus graph. A block of $G$ is called an \textit{end block} of $G$ if it contains exactly one cut vertex of $G$. Note that an end block of a cactus graph $G$ is also a bad block of $G$. Bad blocks of a cactus graph $G$ have another characterization which we state in the following observation.
\begin{Ob}\label{Ob:cactus_bad}
A block $B$ of a cactus graph $G$ is bad if and only if $B$ does not contain two adjacent cut vertices of $G$.
\end{Ob}

If $B_i$ and $B_j$ are two blocks of a block/cactus graph $G$ and $V(B_i) \cap  V(B_j) \neq \emptyset$, then $V(B_i) \cap  V(B_j) = 1$ and the vertex $x \in V(B_i) \cap  V(B_j)$ is a cut vertex of $G$. Let $T$ be a MIST of a nontrivial block/cactus graph $G$. Below, we state two observations which hold true for both block and cactus graphs.

\begin{Ob}
\label{Ob:blockcactus1}
$T$  must have at least one leaf in every bad block of $G$.
\end{Ob}

\begin{Ob}
\label{Ob:blockcactus2}
$Opt(G) \leq n - \vert \text{Bad}(G) \vert$, where $Opt(G)$ denotes the number of internal vertices in $T$.
\end{Ob}

Recall that block decomposition of a graph $G$ is the set of blocks of $G$. It can be computed in $O(n)$ time using the following approach. Let $b$ be a cut vertex of a block/cactus graph $G$ and $G_1, G_2, \ldots, G_t$ be the connected components of the graph $G-b$. Let $H_i$ denotes the subgraph $G[V(G_i) \cup \{b\}]$, for each $1 \leq i \leq t$. We call $H_1, H_2, \ldots, H_t$ the $b$-components of $G$. The block decomposition of a block/cactus graph can be found by recursively choosing a cut vertex $b$ and computing the $b$-components.
\subsection{Algorithm for Block and Cactus Graphs}
\label{blockcactus}
In this subsection, we first prove a theorem which relates the number of internal vertices in a MIST of a block/cactus graph $G$ to the number of bad components of $G$. Then, we outline a linear-time algorithm to compute a MIST of $G$.
\begin{theorem}
	\label{thm:blockcactus}
	Let $G$ be a graph with a nontrivial block decomposition such that each block has a spanning path with a cut vertex as an endpoint. Then $G$ has a spanning tree $T$ in which number of internal vertices is $n - \vert Bad(G)\vert$.
\end{theorem}
\begin{proof}
	Let $l$ be the number of blocks in $G$ and $B_i \in B(G)$ be an arbitrary block of $G$. If $B_i$ is good, then let $P_i$ be a spanning path between two cut vertices of $B_i$. If $B_i$ is bad, we let $P_i$ be a spanning path with a single cut vertex as an endpoint. Let $T = \bigcup_{i = 1}^{l} P_i$. Note that $T$ is a spanning tree of $G$. Furthermore, as any cut vertex of $G$ cannot be a leaf of $T$, we have $i(T) = n - \vert Bad(G)\vert$.
\end{proof}

The proof of Theorem \ref{thm:blockcactus} gives a simple algorithm for a block or cactus graph. First find a block decomposition, this takes $O(n)$ time. Then for each block $B$, determine if $B$ is bad or not and find the corresponding path. This takes $O(\vert B\vert)$ time. In total we have a linear-time algorithm. As both block and cactus graphs satisfy the hypothesis of Theorem \ref{thm:blockcactus}, combining with Obervation \ref{Ob:blockcactus2} we have the following,

\begin{corollary}\label{cor:block}
	If $G$ is a block or cactus graph, then $Opt(G) = n - \vert Bad(G)\vert$.
\end{corollary}

\subsection{Relationship between $Opt(G)$ and $\vert E(P^*)\vert $}
\label{relationship}

We now show for a block or cactus graph $G$ there does not exist a constant $k$ such that $Opt(G) \ge \vert E(P^*)\vert - k$ where $P^*$ is an optimal path cover of $G$.  Recall Corollary \ref{cor:block} states that $Opt(G) = n - \vert Bad(G)\vert$ and Theorem \ref{thm:general upper bound} states $Opt(G) \le \vert E(P^*)\vert -1$.
Note that the number of edges in the optimal path cover $P^*$ and the number of components in $P^*$ adds up to $n$. So, we see that $n - \vert Bad(G)\vert = \vert E(P^*)\vert -(\vert Bad(G)\vert-\vert P^*\vert)$. Thus, $Opt(G) = \vert E(P^*)\vert -(\vert Bad(G)\vert- \vert P^*\vert)$ for both block and cactus graphs.

For every integer $n = 5k ~ ( k \geq 1)$, we construct a connected  graph $G_n$ with $n$ vertices and $Opt(G_n) = \vert E(P^*)\vert - O(n) $. The graph $G_n$ is both a block graph and a cactus graph as every block of $G_n$ is either an edge or a clique on three vertices. The vertex set of $G_n$  is $V(G_n) = V_1 \cup V_2 \cup \ldots \cup V_k$, where $V_i = \{x_1^i, x_2^i, \ldots, x_5^i\}$ for each $ i \in \{1,2, \ldots, k\}$. The edge set is $E(G_n) = E_1 \cup E_2 \cup \ldots \cup E_k \cup E'$, where $E_i = \{x_1^i x_2^i, x_2^i x_3^i, x_3^i x_1^i, x_3^i x_4^i, x_4^i x_5^i, x_5^i x_3^i\}$ for each $i$ and $E'$ contains the edges of the form $x_3^i x_3^{i+1}$ for $1 \leq i \leq (k-1)$. Note $\vert E(G_n)\vert = 7k -1$. We obtain an optimal path cover $P^*$ for $G_n$ having $4k$ edges and $k$ components \cite{pak}. The number of bad blocks in $G_n$ is $2k$. Using Theorem \ref{thm:blockcactus}, we obtain a MIST $T$ of $G_n$ with $n-\vert Bad(G)\vert=5k-2k=3k$ internal vertices. Thus, $Opt(G_n) = 3k = 4k - k = 4k - \frac{n}{5} = \vert E(P^*)\vert - O(n)$. Fig.~\ref{fig:bcbound} provides an illustration for $G_{20}$.

\begin{figure}[h!]
 \centering
 \includegraphics[width=10cm, height=6cm]{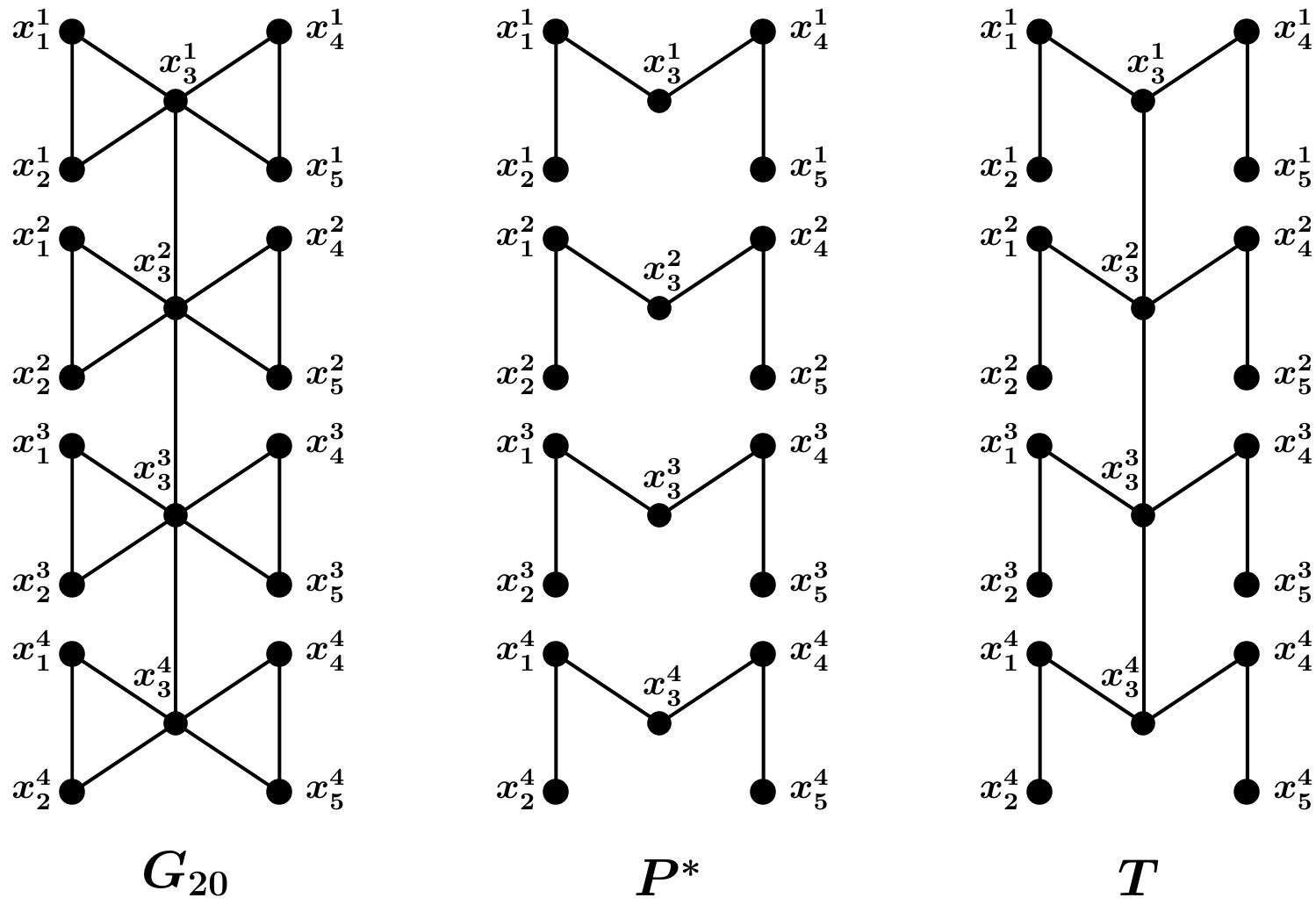}\\
 \caption{Graph $G_{20}$, its optimal path cover $P^*$ and its MIST $T$}
\label{fig:bcbound}
\end{figure}
 Here, we see that $\vert Bad(G_n)\vert- \vert P^*\vert=2k-k=k$ which implies that for arbitrary $n=5k$, we have $Opt(G_n) = \vert E(P^*)\vert -k$. So, block and cactus graphs do not have lower bound for $Opt(G)$ of the form $\vert E(P^*)\vert - c$ for some fixed natural number $c$, independent of $n$.
\section{Cographs}
\label{sec:cograph}
In this section, we discuss the MIST problem for cographs. The complement-reducible graphs or \textit{cographs}  have been discovered independently by several authors since the 1970s \cite{sein, jung}. In the literature, the cographs are also known as $P_4$-free graphs, $D^*$-graphs, Hereditary Dacey graphs and $2$-parity graphs. The class of cographs is defined recursively as follows:
\begin{itemize}
    \item  A single-vertex graph is a cograph;
    \vspace*{1 mm}
	\item If $G$ is a cograph, then its complement $\overline{G}$ is also a cograph;
	 \vspace*{1 mm}
	\item If $G$ and $H$ are cographs, then their disjoint union is also a cograph.
\end{itemize}
Cographs admit a rooted tree representation. This tree is called \textit{cotree} of a cograph $G$, denoted $T(G)$. The cotree of a cograph rooted at a node $r$ possesses the following properties.
\begin{enumerate}
\item Every internal vertex except $r$ has at least two children. Furthermore, $r$ has exactly one child if and only if the underlying cograph $G$ is disconnected.
\item The internal nodes are labeled by either $0$ ($0$-nodes) or $1$ ($1$-nodes) in such a way that the root is always a $1$-node, and $1$-nodes and $0$-nodes alternate along every path in $T(G)$ starting at the root.
\item The leaves of $T(G)$ are precisely the vertices of $G$, such that vertices $x$ and $y$ are adjacent in $G$ if, and only if, the lowest common ancestor of $x$ and $y$ in $T(G)$ is a $1$-node.
\end{enumerate}
Fig.~\ref{fig:cograph} illustrates a cograph $G$ along with its cotree $T(G)$.
\begin{figure}[h!]
 \centering
 \includegraphics[width=8cm]{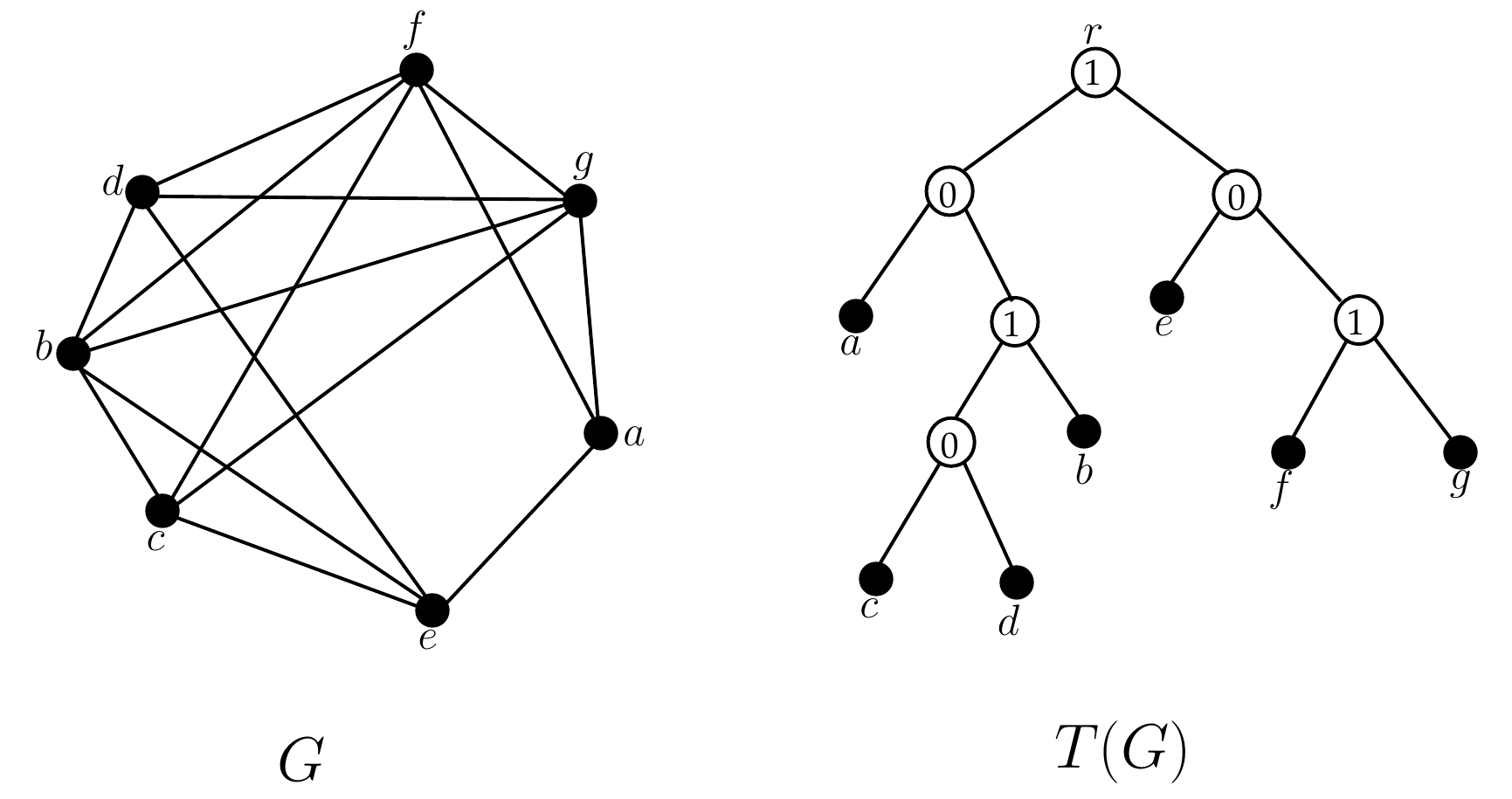}
 \caption{Illustrating a cograph and its cotree}
\label{fig:cograph}
\end{figure}

According to \cite{lin}, the cotree of any cograph $G$ can be preprocessed such that it is a binary tree. So, we may assume that $T(G)$ is a binary tree rooted at a vertex $r$. The set of leaves of the left subtree  of  $T(G)$ is denoted by $L(r_{left})$ and the set of leaves of the right subtree of $T(G)$ is denoted by $L(r_{right})$.

Note that every leaf of $T(G)$ represents a vertex of the graph $G$. If we consider one vertex from $L(r_{left})$ and one vertex from $L(r_{right})$ then their least common ancestor is the root node. As the root node is always a $1$-node, we have the following observation.

\begin{Ob}
\label{Ob:cograph}
For any $x \in L(r_{left}), y \in L(r_{right})$, we have $xy \in E(G)$.
\end{Ob}

Recall that a path cover $P$ of a graph $G$ is a spanning subgraph such that every component of $P$ is a path. A path cover is an optimal path cover if it has the maximum number of edges. \cite{lin} gave a linear-time algorithm to compute an optimal path cover  of a cograph $G$. The optimal path cover $P^*$ constructed in \cite{lin} is one of the following type:
\begin{itemize}
	\item The path cover $P^*$ contains a single path component which is a Hamiltonian path of $G$.\\
	\item The path cover $P^*$ contains at least two path components. In this case, there exists exactly one path $p$ in $P^*$ which contains vertices from both the sets $L(r_{left})$ and $L(r_{right})$ and all other paths in $P^*$ contain vertices from $L(r_{right})$ only. Fig.~\ref{fig:cograph_MIST} illustrates this case.
\end{itemize}

\begin{figure}[h!]
 \centering
 \includegraphics[width=12cm, height=2cm]{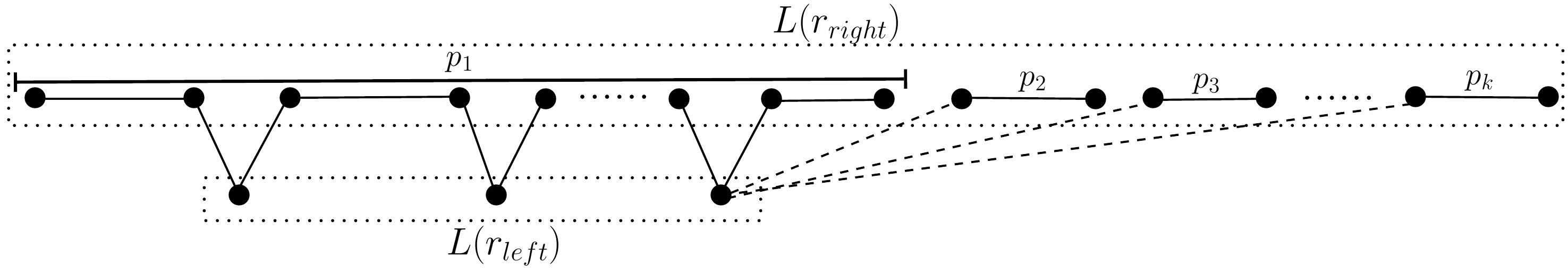}
 \caption{Optimal path cover of $G$ contains more than $1$ path components}
\label{fig:cograph_MIST}
\end{figure}
%
Algorithm \ref{algo:cographMIST} uses the optimal path cover constructed from \cite{lin} to compute a MIST of a cograph $G$. \\
\begin{algorithm}[ht!]
\caption{\label{algo:cographMIST} Algorithm for finding a MIST of a cograph $G$}
  \textbf{Input:} {A cograph $G$ and a cotree $T(G)$ of $G$}
  \textbf{Output:} {A Maximum Internal Spanning Tree $T$ of $G$}
  \nl Let $P^* = \{P_1, P_2, \ldots, P_k \}$ be the optimal path cover of $G$ computed by the algorithm in \cite{lin};

 \nl   $V(T) = V(G)$ and $E(T) = E(P^*) $;

 \nl    \If{ $k=1$}{
 \nl             return $T$;}
   \nl   \Else{
       {/* $P_1$ is the path which contains vertices from both the sets $L(r_{left})$ and $L(r_{right})$ and all other paths in $P^*$ contain vertices from $L(r_{right})$ only */}\\
                \nl  Let $u \in (V(P_1) \cap L(r_{left}))$;\\
             \nl     Let $v_i$ be an end vertex of the path $P_i$, for $2 \leq i \leq k$;\\
            \nl      $E(T) = E(T) \cup \{uv_2, uv_3, \ldots, uv_k\}$;\\
              \nl    return $T$.}
\end{algorithm}

Note that by Theorem \ref{thm:general upper bound} we have $Opt(G) \le \vert{E(P^*)}\vert - 1$ for an optimal path cover $P^*$. Below, we give a theorem which implies that Algorithm \ref{algo:cographMIST} also outputs a spanning tree which attains this upper bound.

\begin{theorem}
\label{thm:cographs}
Algorithm \ref{algo:cographMIST} outputs a spanning tree $T$ of a cograph $G$ such that, $i(T) = \vert E(P^*)\vert-1$, where $P^*$ is an optimal path cover of $G$. Hence, $Opt(G) = \vert E(P^*)\vert-1$.
\end{theorem}
\begin{proof}
Let $P^* = \{P_1, P_2, \ldots, P_k\}$ be the optimal path cover computed in step 1 of Algorithm \ref{algo:cographMIST}. If $\vert P^*\vert = 1$, then $G$ has a Hamiltonian path and Algorithm \ref{algo:cographMIST} returns a Hamiltonian path. Now, suppose $\vert P^*\vert > 1$, then the path $P_1$ contains vertices from both sets $L(r_{left})$ and $L(r_{right})$ and $P_i \cap L(r_{left}) = \emptyset$ for all $i \geq 2$. Now, let $u \in V(P_1) \cap L(r_{left})$ such that $u$ is not an end vertex of $P_1$.

For each path in $P_i \in P^* \setminus \{P_1\}$, consider a pendent vertex $v_i$ of the path. By Observation \ref{Ob:cograph}, $v_i$ and $u$ are adjacent. Let $T = \bigcup_{i = 1}^k P_i \cup \{v_iu : 2 \le i \le k\}$. These new edges connect one internal vertex with a pendant vertex of path of $P^*$. This is illustrated by the dash edges in Fig.~\ref{fig:cograph_MIST}. Note then the number of internal vertices of $T$ is $\vert E(P^*)\vert - 1$, hence $i(T) = Opt(G) = \vert E(P^*)\vert - 1$ by Theorem \ref{thm:general upper bound}.
\end{proof}

Note that step 1 of Algorithm \ref{algo:cographMIST} can be performed in linear-time \cite{lin}. Furthermore, note that the construction of $T$ in Theorem \ref{thm:cographs} is also linear-time. Therefore Algorithm \ref{algo:cographMIST} outputs a MIST of $G$ in linear-time.
\section{Bipartite Permutation Graphs}
\label{sec:BP}
In this section, we discuss the MIST problem for bipartite permutation graphs.
A graph $G = (V,E)$ with $V = \{v_1, v_2, \ldots, v_n\}$ is said to be a \textit{permutation graph} if there is a permutation $\sigma$ over $\{1,2, \ldots, n\}$ such that $v_i v_j \in E$ if and only if $(i-j) (\sigma(i) - \sigma(j)) < 0$. Intuitively, each vertex $v$ in a permutation graph corresponds to a line segment $l_v$ joining two points on two parallel lines $L_1$ and $L_2$,  which is called a line representation. Then, two vertices $v$ and $u$ are adjacent if and only if the corresponding line segments $l_v$ and $l_u$ are crossing. Vertex indices give the ordering of the points on $L_1$, and the permutation of the indices gives the ordering of the points on $L_2$. When a permutation graph is bipartite, it is said to be a \textit{bipartite permutation graph}.

A strong ordering $(<_X, <_Y)$ of a bipartite graph $G = (X,Y,E)$ consists of an ordering $<_X$ of $X$ and an ordering $<_Y$ of $Y$, such that for all edges $ab, a'b'$ with $a, a' \in X$ and $b, b' \in Y$: if $a <_X a'$ and $b' <_Y b,$ then $ab'$ and $a'b$ are edges in $G$. An ordering $<_X$ of $X$  has the \textit{adjacency property} if, for every vertex in $Y$, its neighbors in $X$ are consecutive in $<_X$. The ordering $<_X$ has the \textit{enclosure property} if, for every pair of vertices $y, y'$ of $Y$ with $N(y) \subseteq N(y')$, the vertices of $N(y') \setminus N(y)$ appear consecutively in $<_X$. These properties are useful for characterizing bipartite permutation graphs.

\cite{hegg} proved that \textit{a bipartite graph is a bipartite permutation graph if and only if it admits a strong ordering}. Furthermore if we assume that the graph is connected, then we can say more.

\begin{lemma}\cite{hegg}\label{lemma:strongbpg}
Let $(<_X , <_Y)$ be a strong ordering of a connected bipartite permutation graph $G = (X, Y, E)$. Then both $<_X$ and $<_Y$ have the adjacency property and the enclosure property.
\end{lemma}

Throughout this section, $G=(X,Y,E)$ denotes a connected bipartite permutation graph. A strong ordering of a bipartite permutation graph can be computed in linear-time \cite{spin}. Let $(<_X , <_Y)$ be a strong ordering of $G$, where $<_X = (x_1,x_2, \ldots, x_{n_1})$ and $<_Y = (y_1,y_2, \ldots, y_{n_2})$. We write strong ordering of vertices of $G$ as $(<_X, <_Y) = (x_1, x_2, \ldots, x_{n_1}, y_1, y_2, \ldots, y_{n_2})$. For $u,v \in V(G)$, we write $u <_X v$ if $u,v \in X$ and $u$ appears before $v$ in the strong ordering; we define $u <_Y v$ in a similar manner. We write $x_i < x_j$ (or, $y_i < y_j$) when $i<j$. For vertices $u,v$ of $G$, $u \leq v$  denotes either $u <_X v$, $u <_Y v$, or $u = v$ holds.

Since each vertex of $G$ satisfies the adjacency property, the neighborhood of any vertex  consists of consecutive vertices in the strong ordering. We define the first neighbor of a vertex as the vertex with minimum index in its neighborhood and the last neighbor of a vertex as the vertex with maximum index in its neighborhood. We notate the first and last neighbors of a vertex $u$ as $f(u)$ and $l(u)$ respectively. Combining the above statements for a bipartite permutation graph $G$ with its strong ordering $(<_X , <_Y )$, $G$ has the following properties \cite{lai19}:
\begin{enumerate}
	\item  For any vertex of $G$, its neighborhood consists of consecutive vertices in $<_X$ or $<_Y$ .
	\item For any pair of vertices $u, v$ from $X$ or $Y$, if $u < v$ then $f(u) \leq f(v)$ and $
	l(u) \leq l(v)$.
\end{enumerate}

Now, we define some terminology which we require for the remainder of this section.
A vertex $x_i \in X, (1 \leq i \leq n_1)$ with $l(x_i) = y_j$ is of \textit{type 1}  if $j \geq i$. A vertex $y_i \in Y, (1 \leq i \leq n_2)$ with $l(y_i) = x_j$ is of \textit{type 1}  if $j \geq i+1$. Similarly, a vertex $x_i \in X, (1 \leq i \leq n_1)$ with $l(x_i) = y_j$ is of \textit{type 2}  if $j \geq i+1$ and a vertex $y_i \in Y, (1 \leq i \leq n_2)$ with $l(y_i) = x_j$ is of \textit{type 2}  if $j \geq i$. Note that a type 2 vertex $x \in X$ is also a type 1 vertex but the converse may not be true. Furthermore, a type 1 vertex $y \in Y$ is also a type 2 vertex. Characterizing the vertices in this way is an important distinction for our algorithm. We now prove two important lemmas which will be used to prove the correctness of Algorithm \ref{algo:bpgMIST}.

\begin{lemma}\label{lemma:case1bpg}
Let $X' = \{x_1, x_2, \ldots, x_k, x_{k+1}\} \subseteq X,~~Y' = \{y_1, y_2, \ldots, y_k\} \subseteq Y$. Furthermore, suppose  each vertex of $X'$ and $Y'$ is of type 1 except $x_{k+1}$, $l(x_{k+1}) = y_k$ and $N(X') = Y'$. Then there exists a MIST $T$ of $G$, in which $x_1$ are $x_{k+1}$ are pendant. Moreover; if $X \setminus X' \neq \emptyset$, then the support vertex of $x_{k+1}$ is of degree at least $3$ in $T$.
\end{lemma}
\vspace*{-.5 cm}
\begin{proof}
We first show $x_iy_i,y_ix_{i+1} \in E(G)$ for all $1 \le i \le k$. Suppose there exists $1 \le i \le k$ such that $x_iy_i \not \in E(G)$. Let $l(x_i) =y_j$ and $l(y_i) =x_l$. As both $x_i$ and $y_i$ are type 1, we have $y_j \geq y_i$ and $x_l > x_i$. As $(<_X,<_Y)$ is a strong ordering, we have $x_iy_i \in E(G)$, a contradiction. Thus we may assume $x_iy_i \in E(G)$. Furthermore as $y_i$ is type 1, we have $x_iy_i,y_ix_{i+1} \in E(G)$ for all $1 \le i \le k$.
	
Suppose $X =  X'$. Note that as $(<_X,<_Y)$ is a strong ordering of $G$, we have for all $x \in X$, $l(x) \le l(x_{k+1}) = y_k$. As we assumed $G$ is connected, we have that $Y = Y'$ as well. Note that this implies that $G$ has the Hamiltonian path $x_1y_1x_2 \ldots x_ky_kx_{k+1}$ which is a MIST.  So, we may assume that $X \setminus X' \neq \emptyset$.

Now, let $T^*$ be a MIST of $G$.  If $x_1$ and $x_{k+1}$ are pendant in $T^*$ and degree of $S(x_{k+1})$ is at least $3$ in $T^*$, then we are done. Suppose otherwise, and we modify $T^*$ in the following way. We first remove all edges of $T^*$ incident with the vertices of $X'$ and then add edges  $x_1 y_1, y_1 x_2, x_2 y_2, \ldots, x_k y_k \text{~and~} y_k x_{k+1}$ to obtain a new graph $T$. Note that as $N(X') = Y'$, $T$ is connected.

First suppose $T$ contains no cycle. Note that $T$ is a spanning tree of $G$. If $d_{T}(y_k) =2$, then as $N(X') = Y'$ we can choose an edge $v y_i (i<k)$ in $T$ such that $v \in X \setminus X'$. Since the strong ordering $(<_X,<_Y)$ of the vertices of $G$ satisfies property $2$, we have $vy_k \in E(G)$. So we can further modify $T$ by removing the edge $vy_i$ and replacing with the edge $vy_k.$ We see that
\begin{align*}
i(T^*) &= i_{T^*}(X') + i_{T^*}(X \setminus X') + i_{T^*}(Y') + i_{T^*}(Y \setminus Y')\\
	 & \leq (k-1) +  i_{T^*}(X \setminus X') + k + i_{T^*}(Y \setminus Y')= i(T).
\end{align*}

So, we have $i(T^*) \leq i(T)$. Since $T$ is a spanning tree and $T^*$ is a MIST of $G$, we have that $T$ is also a MIST. Thus, we obtain our desired MIST in which $x_1$ are $x_{k+1}$ are pendant and the support vertex of $x_{k+1}$ is of degree at least $3$.

Now, suppose $T$ contains a cycle $C$. This implies that there exists $v \in X \setminus X'$ such that $vy_i, vy_j  \in E(C)$ with $i<j \leq k$. Now, we modify $T$ by removing the edge $vy_i$. This step reduces degree of $v$ by $1$ while leaving the graph $T$ connected. We repeat this modification until $T$ has no more cycles, thus $T$ will be a spanning tree of $G$. Let us assume that there are  $\alpha$ such vertices which became pendant in this process of updating $T$.
Let  $A \subseteq X \setminus X'$ be the set of $\alpha$ vertices. Note these $\alpha$ vertices were internal in $T^*$.
Suppose $i_{T^*}(X') > k - (\alpha + 1)$. As $N(X') = Y'$, then the subforest of $T^*$ induced by the set $X' \cup Y' \cup A$ would have at least $2(k - \alpha) + (1 + \alpha) + 2\alpha = 2k + 1 + \alpha$ edges. As $2k + 1 + \alpha > \vert X' \cup Y' \cup A \vert - 1$, this contradicts the fact that $T^*$ was a tree. Thus we have $i_{T^*}(X') \le k - (\alpha + 1)$. It follows,
\begin{align*}
	i(T^*) &= i_{T^*}(X') + i_{T^*}(X \setminus X') + i_{T^*}(Y') + i_{T^*}(Y \setminus Y')\\
	&\leq (k-(\alpha +1)) +  i_{T^*}(X \setminus X') + k + i_{T^*}(Y \setminus Y')\\
	&= (k-1) + (i_{T^*}(X \setminus X') - \alpha) + k + i_{T^*}(Y \setminus Y') =i(T).
\end{align*}

Again, we have $i(T^*) \leq i(T)$ which implies that $T$ is also a MIST. If $d_T(y_k) =2$, then we can choose an edge $v y_i (i<k)$ in $T$, such that $v \in X \setminus X'$. Since the strong ordering $(<_X,<_Y)$ satisfies property $2$, we have  $vy_k \in E(G)$. So we update the tree $T$ by removing the edge $vy_i$ and adding the edge $vy_k$. Thus, we obtain a MIST $T$ in which $x_1$ are $x_{k+1}$ are pendant and support vertex of $x_{k+1}$ is of degree at least $3$.
\end{proof}
\begin{lemma}\label{lemma:case2bpg}
Let $X' = \{x_1, x_2, \ldots, x_k\} \subseteq X,~~Y' = \{y_1, y_2, \ldots, y_k\} \subseteq Y$. Furthermore, suppose each vertex of $X'$ and $Y'$ is of type 1 except $y_{k}$, $l(y_k) = x_k$ and $N(Y') = X'$.
\begin{enumerate}
\item[(a)] If $x_iy_{i+1} \in E(G)$  for all $1 \leq i \leq  (k-1)$, then there exists a MIST $T$ of $G$, in which $y_1$ is pendant.
\item[(b)] If there exists $1 \leq t \leq  (k-1)~\text{such that~} x_ty_{t+1} \notin E(G)$, then there exists a MIST $T$ of $G$, in which $x_1$ and $y_{k}$ are pendant. Moreover; if $Y \setminus Y' \neq \emptyset$, then support vertex of $y_k$ is of degree at least $3$ in $T$.
\end{enumerate}
\end{lemma}
\vspace*{-.5 cm}
\begin{proof}
	
We first argue that $x_iy_i,y_ix_{i+1} \in E(G)$ for $1 \le i \le k-1$. First assume for some $i$, $x_iy_i \not \in E(G)$. As both $x_i$ and $y_i$ are type 1, we have $x_i < l(y_i)$ and $y_i < l(x_i)$. As $(<_X,<_Y)$ is a strong ordering, we have $x_iy_i \in E(G)$, a contradiction. Furthermore, as $y_i$ is type 1, we have $y_ix_{i+1} \in E(G)$ for all $1 \le i \le k-1$.

Suppose $Y' = Y$. As $N(Y') = X'$, and $G$ is connected we have that $X = X'$ as well. Note then if $x_iy_{i +1} \in E(G)$ for all $1 \le i \le (k-1)$, then $y_1x_1 \ldots x_{k-1}y_kx_k$ is a Hamiltonian path. Otherwise, if there exists  $1 \le t \le (k-1)~\text{such that~} x_ty_{t+1} \notin E(G)$, then the path $x_1y_1,y_1x_2,x_2y_3, \ldots,y_{k-1}x_k$ and $x_ky_k$ gives the desired Hamiltonian path.

 So, we may assume that $Y \setminus Y' \neq \emptyset$ and we will first prove part (a). Let $T^*$ be a MIST of $G$ and suppose $y_1$ is not pendant in $T^*$. Let $T$ be the graph obtained from $T^*$ where we remove all edges incident to the vertices of $Y'$ and add edges $y_1 x_1 , x_1 y_2, y_2 x_2, \ldots, x_{k-1} y_k \text{~and~} y_k x_{k}$. Note as $N(Y') = X'$, $y_1$ is pendant in $T$.

First, suppose $T$ contains no cycles. Note then that $T$ is a spanning tree of $G$. We argue that we may assume $d_T(x_k) \ge 2$. Suppose otherwise, that is, $d_T(x_k) = 1$. Let $v \in Y \setminus Y'$ such that $vx_i (i<k)$.  As the strong ordering of the vertices of $G$ satisfies property 2, we have $vx_k \in E(G)$ as well. So we further modify $T$ by removing the edge $vx_i$ and adding the edge $vx_k$. We see that
\begin{align*}
	i(T^*) &= i_{T^*}(X') + i_{T^*}(X \setminus X') + i_{T^*}(Y') + i_{T^*}(Y \setminus Y')\\
	&\leq k +  i_{T^*}(X \setminus X') + (k-1) + i_{T^*}(Y \setminus Y') = i(T).
\end{align*}

So, we have that $i(T^*) \leq i(T)$. Since $T$ is a spanning tree and $T^*$ is a MIST of $G$, $T$ is also a MIST of $G$.

Next, suppose $T$ is not a tree. We now modify $T$ to remove the cycles. Let $C$ be a cycle of $T$. Note then there is a vertex $v \in Y \setminus Y'$ such that $vx_i, vx_j \in E(C)$ with $i < j \leq k$. We then modify $T$ by removing the edge $vx_i$. Note that the degree of $v$ decreases by 1. We repeat this process until no cycles remain in $T$. Assume that $\alpha$ cycles were removed during this process and thus at most $\alpha$ pendant vertices were created in this process. As $N(Y') = X'$ and $T^*$ is a tree, we have $i_{T^*}(Y') \leq k-\alpha-1$. We see that,
\begin{align*}
	i(T^*) &= i_{T^*}(X') + i_{T^*}(X \setminus X') + i_{T^*}(Y') + i_{T^*}(Y \setminus Y')\\
	&\leq k +  i_{T^*}(X \setminus X') + (k-\alpha -1) + i_{T^*}(Y \setminus Y')\\
	&\leq k +  i_{T^*}(X \setminus X') + (k -1) + (i_{T^*}(Y \setminus Y') -\alpha) = i(T)
\end{align*}

Again, we have that $i(T^*) \leq i(T)$ which implies that $T$ is also a MIST. Hence, part (a) holds.

Next, we prove part (b). Let $T^*$ be a MIST of $G$. If $x_1$ and $y_k$ are pendant in $T^*$ and degree of $S(y_k)$ is at least $3$ in $T^*$, then we are done, so assume otherwise.
Let $T$ be the graph obtained from modifying $T^*$ where we  remove all edges incident on the vertices of $Y'$ and add edges $x_1 y_1 , y_1 x_2, x_2 y_2, \ldots, y_{k-1} x_k$ and $x_k y_{k}$.

First suppose $T$ contains no cycle, then $T$ is a spanning tree of $G$. If $d_T(x_k) \geq 3$, then we are done modifying, so suppose $d_T(x_k) = 2$. As $Y \setminus Y' \neq \emptyset$ and $N(Y') = X'$, there exists an edge $vx_i (i<k)$ in $T$. Since the strong ordering of the vertices of $G$ satisfies property 2, we have $vx_k \in E(G)$. Thus we further modify $T$ where we remove $vx_i$ and add the edge $vx_k$. As we assumed there exists a $1 \le t \le (k -1)$ such that $x_ty_{t+1} \not \in E(G)$, we have $N(\{x_1,x_2,...,x_t\}) = \{y_1,y_2,...,y_t\}$. Let $X'' = \{x_1,x_2, \ldots, x_t\}$ and note $N(X'') = Y'' = \{y_1,y_2,...,y_t\}$. By Observation \ref{Ob:basic}, we see that for any spanning tree of $G$, $X''$ contains at least one pendant vertex. So, $i_{T^*}(X') \leq (k-1)$.We see that
\begin{align*}
i(T^*) &= i_{T^*}(X') + i_{T^*}(X \setminus X') + i_{T^*}(Y') + i_{T^*}(Y \setminus Y')\\
&\leq (k-1) +  i_{T^*}(X \setminus X') + (k-1) + i_{T^*}(Y \setminus Y') = i(T).
\end{align*}

Again, we have $i(T^*) \leq i(T)$ which implies that $T$ is a MIST. Thus, we obtained a MIST $T$ in  which $x_1$ and $y_{k}$ are pendant and support vertex of $y_{k}$ is of degree at least $3$.

Now, suppose $T$ contains a cycle. We now modify $T$ to be a spanning tree. Let $C$ be a cycle contained in $T$. This implies that there is a vertex $v \in Y \setminus Y'$ such that $vx_i, vx_j \in E(C)$ with $i < j \leq k$. We remove the edge $vx_i$ from $T$. This decreases the degree of $v$ by 1. We repeat this process until no cycles remain in $T$. Let $A \subseteq Y \setminus Y'$ with $\vert A\vert = \alpha$ be the set of the vertices made pendant in this process. Suppose $i_{T^*}(Y') \ge (k - \alpha)$. As $N(Y') = X'$, the subgraph of $T^*$ induced by $X' \cup Y' \cup A$ has at least $(2 k- \alpha) + k + 2\alpha = 2k + \alpha$ edges. As $\vert X' \cup Y' \cup A\vert = 2k + \alpha$, this contradicts the fact that $T^*$ is a tree. Thus we may assume $i_{T^*}(Y') \le k - \alpha - 1$. As before, we may assume $d_T(x_k) \ge 3$. It follows,
\begin{align*}
	i(T^*) &= i_{T^*}(X') + i_{T^*}(X \setminus X') + i_{T^*}(Y') + i_{T^*}(Y \setminus Y')\\
	&\leq (k-1) +  i_{T^*}(X \setminus X') + (k-\alpha -1) + i_{T^*}(Y \setminus Y')\\
	&\leq (k-1)
	+  i_{T^*}(X \setminus X') + (k -1) + (i_{T^*}(Y \setminus Y') -\alpha)= i(T)
\end{align*}

This implies that $T$ is also a MIST. Thus, we obtained a MIST $T$ in  which $x_1$ and $y_{k}$ are pendant and support vertex of $y_{k}$ is of degree at least $3$.
\end{proof}
We state similar results when the vertices are of type $2$. By symmetry, the proofs of Lemmas \ref{lemma:case3bpg} and \ref{lemma:case4bpg} follow from Lemmas \ref{lemma:case1bpg} and \ref{lemma:case2bpg}.
\begin{lemma}\label{lemma:case3bpg}
Let $X' = \{x_1, x_2, \ldots, x_k\} \subseteq X,~~Y' = \{y_1, y_2, \ldots, y_k, y_{k+1}\} \subseteq Y$. Furthermore, suppose each vertex of $X'$ and $Y'$ is of type 2 except $y_{k+1}$, $l(y_{k+1}) = x_k$ and $N(Y') = X'$. Then there exists a MIST $T$ of $G$, in which $y_1$ are $y_{k+1}$ are pendant. Moreover; if $Y \setminus Y' \neq \emptyset$, then support vertex of $y_{k+1}$ is of degree at least $3$ in $T$.
\end{lemma}
\begin{lemma}\label{lemma:case4bpg}
Let $X' = \{x_1, x_2, \ldots, x_k\} \subseteq X,~~Y' = \{y_1, y_2, \ldots, y_k\} \subseteq Y$. Furthermore, suppose each vertex of $X'$ and $Y'$ is of type 2 except $x_{k}$, $l(x_{k}) = y_k$ and $N(X') = Y'$.
\begin{enumerate}
\item[(a)] If $y_i x_{i+1} \in E(G) ~\forall ~ 1 \leq i \leq  (k-1)$, then there exists a MIST $T$ of $G$, in which $x_1$ is pendant.
\item[(b)] If $\exists ~ 1 \leq t \leq  (k-1)~\text{such that~} y_t x_{t+1} \notin E(G)$, then there exists a MIST $T$ of $G$, in which $y_1$ and $x_{k}$ are pendant. Moreover; if $X \setminus X' \neq \emptyset$, then support vertex of $x_{k}$ is of degree at least $3$ in $T$.
\end{enumerate}
\end{lemma}
Next, we propose an algorithm to find a MIST of $G$ based on the Lemmas \ref{lemma:case1bpg}, \ref{lemma:case2bpg}, \ref{lemma:case3bpg} and \ref{lemma:case4bpg}. In our algorithm, we first find a vertex $u$ such that it is a pendant vertex in some MIST $T$ of $G$ and the degree of support vertex of $u$ in $T$ is at least $3$. Now, if we remove $u$ from $G$ and call the remaining graph $G'$, then we see that the number of internal vertices in a MIST of $G$ is same as the number of internal vertices in a MIST of $G'$. Note that we can easily construct a MIST of $G$ from a MIST of $G'$ by adding the pendant vertex $u$ to the corresponding support vertex. So, after finding the vertex $u$, the problem is reduced to finding MIST of $G \setminus \{u\}$, say $G'$. We continue doing the same until no such vertex $u$ exists and then the resultant graph has a Hamiltonian path.

In our algorithm, we visit the vertices alternatively from the partitions $X$ and $Y$. We consider two special orderings $(x_1, y_1, x_2, y_2,...)$ and $(y_1, x_1, y_2, x_2,...)$ of $V(G)$ which we call $\alpha$ and $\beta$, respectively.
Below, we describe the method to find a vertex $u$ which is pendant in some MIST $T$ of $G$ and $d_{T}(S(u))$ is at least $3$.

We first visit the vertices of $G$ in the ordering $\alpha$ and search for the first vertex, which is not of type 1. Let $u$ be such a vertex. If $u \in X$ or $u \in Y$ and the conditions of part (b) of Lemma \ref{lemma:case2bpg} are satisfied, then there exists a MIST $T$ of $G$ in which $u$ is a pendant vertex and the degree of support vertex of $u$ in $T$ is at least $3$. So, we remove $u$ from $G$ and find a MIST $T'$ of $G \setminus \{u\}$. Later, we obtain a MIST of $G$ by adding $u$ to $T'$. But, if $u \in Y$, say $u =y_k$ and conditions of part (a) of Lemma \ref{lemma:case2bpg} are satisfied, then there exists a MIST $T$ of $G$ in which $y_1$ is a pendant vertex. In this case, we start visiting the vertices of $G$ in the ordering $\beta$, starting from $y_1$. At this step, we do not maintain any information from $\alpha$ search.

Now, let $u$ be the first vertex not of type 2 in the ordering $\beta$. If $u \in Y$ or $u \in X$ and the conditions of part (b) of Lemma \ref{lemma:case4bpg} are satisfied, then there exists a MIST $T$ of $G$ in which $u$ is a pendant vertex and the degree of support vertex of $u$ in $T$ is at least $3$. So, we remove $u$ from $G$ and find a MIST $T'$ of $G \setminus \{u\}$. Later, we obtain a MIST of $G$ by adding $u$ to $T'$. Here, if $u \in X$ and conditions of part (a) of Lemma \ref{lemma:case4bpg} are satisfied, then there exists a MIST $T$ of $G$ in which $x_1$ is a pendant vertex. But, we have already explored this possibility while visiting the vertices of $G$ in the ordering $\alpha$. So, we do not get such a vertex $u$.  To see this, suppose that we get such a vertex $u$. Then, $u=x_t$ for some $t$, where $t > k$. Now, part (a) of Lemma \ref{lemma:case4bpg} tells that $y_i x_{i+1} \in E(G)$ for all $1 \le i \le (t-1)$ implying that $y_kx_{k+1} \in E(G)$. But, while visiting the vertices in the ordering $\alpha$, we got a vertex $y_k$ satisfying $l(y_k) = x_k$, so $y_k x_{k+1 } \notin E(G)$, a contradiction.

The detailed procedure for computing a MIST of a bipartite permutation graph is presented in Algorithm \ref{algo:bpgMIST}. Algorithm \ref{algo:bpgMIST} either finds a vertex which is not of type 1 or a vertex which is not of type 2. When such a vertex $u$ is found, we call $u$ as an \textit{encountered} vertex. All the encountered vertices are found while executing the steps written in lines 4, 11, 17, 22, 31 or 39 of Algorithm \ref{algo:bpgMIST}. We see that the algorithm returns a spanning tree $T$ of $G$. Before proving the correctness of the Algorithm \ref{algo:bpgMIST}, we state a necessary lemma.

%

\begin{algorithm}
\caption{\label{algo:bpgMIST}\textbf{Algorithm for finding a MIST of a bipartite permutation graph G}}
\small
 \textbf{Input:}{ A bipartite permutation graph $G$ and a strong ordering $(<_X, <_Y) = (x_1,x_2, \ldots, x_{n_1}, y_1, y_2, \ldots, y_{n_2})$ of $V(G)$}.
  \textbf{Output:}{ A MIST $T$ of $G$.}\\
 \setcounter{AlgoLine}{0}
\nl $V(T)=X \cup Y, E(T) = \emptyset, t=0$; $flag =1$;\\
\nl $\alpha = (x_1, y_1, x_2, y_2,...)$ and $\beta = (y_1, x_1, y_2, x_2,...)$;\\
\nl Visit the vertices of $V(G)$ in the ordering $\alpha$; \\
\nl Let $u$ be the first vertex with minimum index in the ordering $\alpha$ which is not of type 1;\\
\nl \While{$flag ==1$}{
\nl \If{$u \in X$}{
\nl         Let $u = x_{k+1}$ for some $k$;\\
\nl           \If{$k+1 \neq n_1$}{
\nl             $t=t+1$; rename $x_{k+1}$ as $a_t$;  $E(T) = E(T) \cup \{y_k a_t\}$;\\
\nl              Rename $x_i$ as $x_{i-1}$ for every $k+2 \leq i \leq n_1 $; $n_1=n_1 - 1$;\\
\nl         Continue looking for a next vertex which is not of type 1 in the ordering $\alpha$,
            call it $u$;}
\nl         \Else{
 \nl         $E(T) = E(T) \cup \{x_1 y_1, y_1 x_2, x_2 y_2, \ldots, x_k y_k, y_k x_{k+1} \}$; return $T$;
        }
   }
\nl \Else{
\nl     Let $u = y_k$ for some $k$;\\
\nl     \If{$x_iy_{i+1} \in E(G) ~\forall ~ 1 \leq i \leq (k-1)$}{
\nl        Find a vertex which is not of type 2 in the ordering $\beta$ starting from $y_1$, call it $u$;  $flag = 2$;
       }
\nl      \Else{
\nl          \If{$k \neq n_2$}{
\nl         $t=t+1$; rename $y_k$ as $a_t$; $E(T) = E(T) \cup \{x_k a_t\}$;\\
\nl            Rename $y_i$ as $y_{i-1}$ for every $k+1 \leq i \leq n_2 $;  $n_2=n_2 - 1$;\\
\nl         Continue looking for a next vertex which is not of type 1 in the ordering $\alpha$, call it $u$;
}
\nl         \Else{
\nl            $E(T) = E(T) \cup \{x_1 y_1, y_1 x_2, x_2 y_2, \ldots,  y_{k-1}
                 x_k, x_k y_k \}$; return $T$;
            }
          }
       }
  }
\nl \While{$flag ==2$}{
\nl  \If{$u \in Y$}{
\nl     Let $u = y_{k+1}$ for some $k$; \\
\nl     \If{$k+1 \neq n_2$}{
\nl         $t=t+1$; rename $y_{k+1}$ as $a_t$; $E(T) = E(T) \cup  \{x_k a_t\}$;\\
\nl         Rename $y_i$ as $y_{i-1}$ for every $k+2 \leq i \leq n_2 $;  $n_2=n_2 - 1$;\\
\nl  Continue looking for a next vertex which is not of type 2 in the ordering $\beta$, call it $u$;
}
\nl     \Else{
\nl       $E(T) = E(T) \cup \{y_1 x_1, x_1 y_2, y_2 x_2, \ldots, y_k x_k, x_k y_{k+1} \}$; return $T$;

        }
    }
\nl     \Else {
\nl               Let $u = x_k$ for some $k$;\\
 \nl                \If{$k \neq n_1$}{
 \nl                $t=t+1$; rename $x_k$ as $a_t$;  $E(T) = E(T) \cup \{y_k a_t\}$;\\
\nl            Rename $x_i$ as $x_{i-1}$ for every $k+1 \leq i \leq n_1 $; $n_1=n_1 - 1$;\\
\nl        Continue looking for a next vertex which is not of type 2 in the ordering $\beta$,  call it $u$;
}
 \nl        \Else{
 \nl           $E(T) = E(T) \cup \{y_1 x_1, x_1 y_2, y_2 x_2, \ldots,
            x_{k-1} y_k, y_k x_k \}$;  return $T$;
            }
          }
       }
\end{algorithm}
%
%

\begin{lemma}\label{lemma:basicstep}
Let $G$ be the input bipartite permutation graph for the Algorithm \ref{algo:bpgMIST} and $a_1$ denotes the first encountered vertex in either the $\alpha$ or $\beta$ search. Suppose that $T$ is the spanning tree of $G$ returned by Algorithm \ref{algo:bpgMIST}. Let $X_1 \subseteq X$ be the set of vertices which are visited from $X$ side till $a_1$ and $Y_1 \subseteq Y$ be the set of vertices which are visited from $Y$ side till $a_1$. Then there exists a MIST $T^*$ of $G$ such that $E(T^*[X_1 \cup Y_1]) = E(T[X_1 \cup Y_1])$.
\end{lemma}
\begin{proof}
We have four cases to consider.

\noindent
\textbf{Case 1}: $a_1 \in X$ and it is not of type 1. Then the  vertex $a_1$ was found when flag = 1 in Algorithm \ref{algo:bpgMIST}, that is, when searching for the first vertex not of type 1. Let $a_1 = x_{k+1}$ for some $k$. Then the sets $X' = \{x_1, x_2, \ldots, x_k, x_{k+1}\} \subseteq X,~~Y' = \{y_1, y_2, \ldots, y_k\} \subseteq Y$ satisfy the hypothesis of Lemma \ref{lemma:case1bpg}.  Thus by Lemma \ref{lemma:case1bpg}, there exists a MIST $T^*$ of $G$ such that $E(T^*[X_1 \cup Y_1]) = \{x_1 y_1, y_1 x_2, x_2 y_2, \ldots, x_k y_k, y_k x_{k+1}\}$. In particular, $E(T^*[X_1 \cup Y_1]) = E(T[X_1 \cup Y_1])$.

\noindent
\textbf{Case 2}: $a_1 \in Y$ and it is not of type 1. Then the vertex $a_1$ was also found when  flag = 1 in the algorithm. Let $a_1 = y_{k}$ for some $k$. Then the sets $X' = \{x_1, x_2, \ldots, x_k\} \subseteq X,~~Y' = \{y_1, y_2, \ldots, y_k\} \subseteq Y$ satisfy the hypothesis of part (b) of Lemma \ref{lemma:case2bpg}. Thus by Lemma \ref{lemma:case2bpg}, there exists a MIST $T^*$ of $G$ such that $E(T'[X_1 \cup Y_1]) = \{x_1 y_1, y_1 x_2, x_2 y_2, \ldots, x_k y_k\} = E(T[X_1 \cup Y_1])$.

By symmetry, the other two cases ($a_1 \in X$ and it is not of type 2; $a_1 \in Y$ and it is not of type 2) follow from Lemmas \ref{lemma:case3bpg} and \ref{lemma:case4bpg}.
Thus there exists a MIST $T^*$ of $G$ such that $E(T^*[X_1 \cup Y_1]) = E(T[X_1 \cup Y_1])$ in all cases.
\end{proof}

Now, we prove the correctness of Algorithm \ref{algo:bpgMIST}.
\begin{theorem}\label{theorem:bpgcorrectness}
Algorithm \ref{algo:bpgMIST} returns a maximum internal spanning tree of $G$.
\end{theorem}
\begin{proof}
Let $T^*$ be a MIST of $G$ and $T$ be the spanning tree of $G$ returned by Algorithm \ref{algo:bpgMIST}. Recall in the execution of Algorithm \ref{algo:bpgMIST}, we either search for a vertex not of type 1 with the ordering $\alpha$ or we search for a vertex not of type 2 with the ordering $\beta$. This is ensured since either we never arrive at line 17 or we arrive at it once and after that flag remains 2 throughout the algorithm. Let $a_1, a_2, \ldots, a_p$ be the sequence of vertices encountered in the execution of Algorithm \ref{algo:bpgMIST}. Let $X_1$ and $Y_1$ denote the set of vertices visited till $a_1$ from $X$ and $Y$ side respectively. For $1 < i < p$, let $X_i$ denotes the set of vertices visited from $X$ side after $a_{i-1}$ and upto $a_i$. Similarly, let $Y_i$ denotes the set of vertices visited from $Y$ side after $a_{i-1}$ and upto $a_i$. Let $X_p$ and $Y_p$ denote the set of all vertices visited after $a_{p-1}$ from $X$ and $Y$ side respectively.

First suppose Algorithm \ref{algo:bpgMIST} is searching for a vertex not of type 1 with the ordering $\alpha$ and it never arrives at line 17. This means that flag is 1 throughout the algorithm. To prove that $T$ is a MIST of $G$, we will prove that $T^*$ can be modified so that it remains a MIST of $G$ and $E(T^*)$ is same as $E(T)$, that is,
\begin{equation}\label{eq}
E(T^*[\bigcup\limits_{j=1}^{p} X_j \cup \bigcup\limits_{j=1}^{p} Y_j]) = E(T[\bigcup\limits_{j=1}^{p} X_j \cup \bigcup\limits_{j=1}^{p} Y_j]).
\end{equation}
 We prove (\ref{eq}) using induction on $p$. If $p=1$, we have $E(T^*[X_1 \cup Y_1]) = E(T[X_1 \cup Y_1])$ due to Lemma \ref{lemma:basicstep}. Hence, (\ref{eq}) is true for $p=1$. Assume that (\ref{eq}) is true for $p=i$.

We now show that (\ref{eq}) is true for $p=i+1$. So, consider vertex $a_{i+1}$ for $i \geq 2$. Two possible cases arise.

\smallskip
\noindent \textbf{Case 1}: $a_{i+1} \in X$.

First suppose $a_j \in X$ for every $1 \leq j \leq i$. Then for $X^* = \bigcup\limits_{j=1}^{i+1} X_j$ and $Y^* = \bigcup\limits_{j=1}^{i+1} Y_j$, we have $\vert X^*\vert = \sum_{j = 1}^{i + 1}\vert X_j\vert$ and $\vert Y^*\vert = \sum_{j = 1}^{i + 1}\vert Y_j\vert = \sum_{j = 1}^{i + 1}(\vert X_j\vert - 1) = \vert X^*\vert - (i+1)$. As $N(X^*) = Y^*,$ by Observation \ref{Ob:basic} we have that the number of pendant vertices from $X^*$ in any spanning tree of $G$ is at least $\vert X^*\vert - \vert Y^*\vert + 1 = i + 2$. Therefore $i_{T^*}(X^*) \le \vert X^*\vert - (i + 2)$.

If $E(T^*[\bigcup\limits_{j=1}^{i+1} X_j \cup \bigcup\limits_{j=1}^{i+1} Y_j]) \neq E(T[\bigcup\limits_{j=1}^{i+1} X_j \cup \bigcup\limits_{j=1}^{i+1} Y_j])$, we remove all edges of $T^*$ who have one end in $\bigcup\limits_{j=1}^{i} (X_j \cup Y_j)$ and the other in $(X_{i+1} \cup Y_{i+1})$ and all edges incident with the vertices of $X_{i+1}$ within $T^*$. We then add all edges from $E(T[X_{i+1} \cup Y_{i+1}])$ and the edge of $T$ which connects $\bigcup\limits_{j=1}^{i} (X_j \cup Y_j)$ to $(X_{i+1} \cup Y_{i+1})$ in $T^*$. If cycles were created in this process, then we can remove those cycles without introducing more pendant vertices using the method discussed in Lemmas \ref{lemma:case1bpg} and \ref{lemma:case2bpg}. Let $T^*_{new}$ denote this updated tree. We have,
\begin{align*}
 	i(T^*) &= i_{T^*}(X^*) + i_{T^*}(X \setminus X^*) + i_{T^*}(Y^*) + i_{T^*}(Y \setminus Y^*)\\
 	&\leq \vert X^* \vert - (i+2)+  i_{T^*}(X \setminus X^*) + \vert Y^* \vert + i_{T^*}(Y \setminus Y^*) = i(T^*_{new})
\end{align*}
Thus $T^*_{new}$ is also a MIST of $G$ such that $$E(T^*_{new}[\bigcup\limits_{j=1}^{i+1} X_j \cup \bigcup\limits_{j=1}^{i+1} Y_j]) = E(T[\bigcup\limits_{j=1}^{i+1} X_j \cup \bigcup\limits_{j=1}^{i+1} Y_j]).$$

Now, suppose $a_j \in Y$ for some  $1 \leq j \leq i$. We choose the largest $j \in \{1,2, \ldots, i\}$ such that $a_j \in Y$. Then for $X^* = \bigcup\limits_{t=j+1}^{i+1} X_t$ and $Y^* = \bigcup\limits_{t=j+1}^{i+1} Y_t$, we have $\vert X^*\vert = \sum_{t = j+1}^{i + 1} \vert X_t\vert$ and $\vert Y^*\vert= \sum_{t = j+1}^{i + 1}\vert Y_t\vert= \vert X_{j+1}\vert+ \sum_{t = j+2}^{i + 1}(\vert X_t\vert-1) = \vert X^* \vert-(i-j)$. As $N(X^*) = Y^*$, by Observation \ref{Ob:basic} we have that the number of pendant vertices from $X^*$ in any spanning tree of $G$ is at least $\vert X^*\vert - \vert Y^*\vert + 1 = i -j+ 1$. Therefore, $i_{T^*}(X^*) \leq \vert X^*\vert - (i-j+1)$.

If $E(T^*[\bigcup\limits_{j=1}^{i+1} X_j \cup \bigcup\limits_{j=1}^{i+1} Y_j]) \neq E(T[\bigcup\limits_{j=1}^{i+1} X_j \cup \bigcup\limits_{j=1}^{i+1} Y_j])$, we remove all edges of $T^*$ whose one end belongs to $\bigcup\limits_{j=1}^{i} (X_j \cup Y_j)$ and another end to $(X_{i+1} \cup Y_{i+1})$, all edges incident on the vertices of $X_{i+1}$ from $T^*$ and add all edges from $E(T[X_{i+1} \cup Y_{i+1}])$ and the edge of $T$ which connects $\bigcup\limits_{j=1}^{i} (X_j \cup Y_j)$ to $(X_{i+1} \cup Y_{i+1})$ in $T^*$. If cycles were created in this process, then we can remove those cycles without introducing more pendant vertices using the method discussed in Lemmas \ref{lemma:case1bpg} and \ref{lemma:case2bpg}. So, we may assume these modifications of $T^*$ do not create any cycles. Let $T^*_{new}$ denotes the updated tree. We have,
$$i(T^*) = i_{T^*}(\bigcup\limits_{t=1}^{j} X_t)+i_{T^*}(X^*) + i_{T^*}(\bigcup\limits_{t=i+2}^{p} X_t) + i_{T^*}(\bigcup\limits_{t=1}^{j} Y_t)+i_{T^*}(Y^*) + i_{T^*}(\bigcup\limits_{t=i+2}^{p} Y_t)$$\\
$$\leq i_{T^*}(\bigcup\limits_{t=1}^{j} X_t) +\vert X^*\vert - (i-j+1) +  i_{T^*}(\bigcup\limits_{t=i+2}^{p} X_t) + i_{T^*}(\bigcup\limits_{t=1}^{j} Y_t)+ \vert Y^*\vert + i_{T^*}(\bigcup\limits_{t=i+2}^{p} Y_t)$$ $= i(T^*_{new}).$\\
Thus, $T^*_{new}$ is also a MIST of $G$ such that $$E(T^*_{new}[\bigcup\limits_{j=1}^{i+1} X_j \cup \bigcup\limits_{j=1}^{i+1} Y_j]) = E(T[\bigcup\limits_{j=1}^{i+1} X_j \cup \bigcup\limits_{j=1}^{i+1} Y_j]).$$
\smallskip
\noindent \textbf{Case 2}: $a_{i+1} \in Y$.

First suppose $a_j \in Y$ for every $1 \leq j \leq i$. Then for $X^* =  \bigcup\limits_{j=1}^{i+1} X_j$ and $Y^* = \bigcup\limits_{j=1}^{i+1} Y_j$, we have $\vert Y^*\vert = \sum_{j = 1}^{i + 1}\vert Y_j\vert$ and $\vert X^*\vert = \sum_{j = 1}^{i + 1}\vert X_j\vert = \vert Y_1\vert + \sum_{j = 2}^{i + 1}(\vert Y_j\vert - 1) = \vert Y^*\vert - i$. As $N(Y^*) = X^*,$ by Observation \ref{Ob:basic} we have that the number of pendant vertices from $Y^*$ in any spanning tree of $G$ is at least $\vert Y^*\vert - \vert X^*\vert + 1 = i + 1$. Therefore $i_{T^*}(Y^*) \le \vert Y^*\vert - (i + 1)$. Here, $a_1 \in Y$, so, let $a_1 = y_k$ for some $k$. As we have assumed that flag is 1, this implies that there exists an index $t,~ 1 \leq t \leq k-1$ such that $x_t y_{t+1} \notin E(G)$. So, for $X'= \{x_1,x_2, \ldots, x_t\}$ and $Y' = \{y_1,y_2, \ldots, y_t\}$, we have $N(X') = Y'$. Now, by Observation \ref{Ob:basic}, we know that the number of pendant vertices within $X'$ in any spanning tree of $G$ is at least $\vert X'\vert -\vert Y'\vert + 1 = 1$. So, $i_{T^*}(X') \leq \vert X'\vert - 1$, implying that $i_{T^*}(X^*) \leq \vert X^*\vert - 1$.

If $E(T^*[\bigcup\limits_{j=1}^{i+1} X_j \cup \bigcup\limits_{j=1}^{i+1} Y_j]) \neq E(T[\bigcup\limits_{j=1}^{i+1} X_j \cup \bigcup\limits_{j=1}^{i+1} Y_j])$, we remove all edges of $T^*$ who have one end in $\bigcup\limits_{j=1}^{i} (X_j \cup Y_j)$ and the other in $(X_{i+1} \cup Y_{i+1})$ and all edges incident with the vertices of $Y_{i+1}$ within $T^*$. We then add all edges from $E(T[X_{i+1} \cup Y_{i+1}])$ and the edge of $T$ which connects $\bigcup\limits_{j=1}^{i} (X_j \cup Y_j)$ to $(X_{i+1} \cup Y_{i+1})$ in $T^*$. As before, if cycles are present, we may further modify $T^*$ to remove these cycles without introducing more pendant vertices. Let $T^*_{new}$ denote this updated tree. We have,
 \begin{align*}
 	i(T^*) &= i_{T^*}(X^*) + i_{T^*}(X \setminus X^*) + i_{T^*}(Y^*) + i_{T^*}(Y \setminus Y^*)\\
 	&\leq \vert X^* \vert - 1 +  i_{T^*}(X \setminus X^*) + \vert Y^*\vert - (i+1) + i_{T^*}(Y \setminus Y^*) = i(T^*_{new})
 \end{align*}
Thus $T^*_{new}$ is also a MIST of $G$ such that $$E(T^*_{new}[\bigcup\limits_{j=1}^{i+1} X_j \cup \bigcup\limits_{j=1}^{i+1} Y_j]) = E(T[\bigcup\limits_{j=1}^{i+1} X_j \cup \bigcup\limits_{j=1}^{i+1} Y_j]).$$

Now, suppose there is some vertex of $X$ side among the vertices $a_j$ for  $1 \leq j \leq i$. We choose the largest $j \in \{1,2, \ldots, i\}$ such that $a_j \in X$. Then for $X^* = \bigcup\limits_{t=j+1}^{i+1} X_t$ and $Y^* = \bigcup\limits_{t=j+1}^{i+1} Y_t$, we have $\vert Y^*\vert = \sum_{t = j+1}^{i + 1}\vert Y_t\vert$ and $\vert X^*\vert = \sum_{t = j+1}^{i + 1}\vert X_t\vert= \vert Y_{j+1}\vert+ \sum_{t = j+2}^{i + 1}(\vert Y_t\vert-1) = \vert Y^*\vert-(i-j)$. As $N(Y^*) = X^*$, by Observation \ref{Ob:basic} we have that the number of pendant vertices from $Y^*$ in any spanning tree of $G$ is at least $\vert Y^*\vert - \vert X^*\vert + 1 = i -j+ 1$. Therefore, $i_{T^*}(Y^*) \leq \vert Y^*\vert - (i-j+1)$.

If $E(T^*[\bigcup\limits_{j=1}^{i+1} X_j \cup \bigcup\limits_{j=1}^{i+1} Y_j]) \neq E(T[\bigcup\limits_{j=1}^{i+1} X_j \cup \bigcup\limits_{j=1}^{i+1} Y_j])$, we remove all edges of $T^*$ who have one end in $\bigcup\limits_{j=1}^{i} (X_j \cup Y_j)$ and the other in $(X_{i+1} \cup Y_{i+1})$ and all edges incident with the vertices of $Y_{i+1}$ within $T^*$. We then add all edges from $E(T[X_{i+1} \cup Y_{i+1}])$ and the edge of $T$ which connects $\bigcup\limits_{j=1}^{i} (X_j \cup Y_j)$ to $(X_{i+1} \cup Y_{i+1})$ in $T^*$. As before, if we created cycles with this modification, we remove them with the same method used in Lemmas \ref{lemma:case1bpg} and \ref{lemma:case2bpg}. Let $T^*_{new}$ denote this updated tree. We have,

$$i(T^*) = i_{T^*}(\bigcup\limits_{t=1}^{j} X_t)+i_{T^*}(X^*) + i_{T^*}(\bigcup\limits_{t=i+2}^{p} X_t) + i_{T^*}(\bigcup\limits_{t=1}^{j} Y_t)+i_{T^*}(Y^*) + i_{T^*}(\bigcup\limits_{t=i+2}^{p} Y_t)$$\\
$$\leq i_{T^*}(\bigcup\limits_{t=1}^{j} X_t) +\vert X^*\vert  +  i_{T^*}(\bigcup\limits_{t=i+2}^{p} X_t) + i_{T^*}(\bigcup\limits_{t=1}^{j} Y_t)+ \vert Y^*\vert -(i-j+1) + i_{T^*}(\bigcup\limits_{t=i+2}^{p} Y_t)$$ $= i(T^*_{new}).$\\
Thus $T^*_{new}$ is also a MIST of $G$ such that $$E(T^*_{new}[\bigcup\limits_{j=1}^{i+1} X_j \cup \bigcup\limits_{j=1}^{i+1} Y_j]) = E(T[\bigcup\limits_{j=1}^{i+1} X_j \cup \bigcup\limits_{j=1}^{i+1} Y_j]).$$

Hence, (\ref{eq}) is true for $p=i+1$, that is, $$E(T^*[\bigcup\limits_{j=1}^{i+1} X_j \cup \bigcup\limits_{j=1}^{i+1} Y_j]) = E(T[\bigcup\limits_{j=1}^{i+1} X_j \cup \bigcup\limits_{j=1}^{i+1} Y_j]).$$
Thus, we get that $E(T^*[X \cup Y] = E(T[X \cup Y])$ in all cases, when flag is 1.

If algorithm arrives at line 17, then flag changes to 2 and it remains 2 throughout the algorithm. So, it searches vertex not of type 2 in the ordering $\beta$ starting from $y_1$. There will be analogous arguments for this case also, using Lemmas \ref{lemma:case3bpg} and \ref{lemma:case4bpg} instead. For a quick justification why,  with the assumption flag = 1, the above analysis fails if  we encounter a vertex, say $u_1 = y_j$, such that $u_1$ is not type 1 and $x_iy_{i+1} \in E(G)$ for all $1 \le i \le (j-1)$. The analogous failure case for the flag = 2 is, when we encounter a vertex $u_2 = x_k$ that is not of type 2 and $y_ix_{i+1} \in E(G)$ for all $1 \le i \le (k - 1)$. Note that these cases cannot simultaneously occur. Otherwise the analysis is symmetric. Consequently, Algorithm \ref{algo:bpgMIST} returns a maximum internal spanning tree of $G$.
\end{proof}

Now, we discuss the running time of Algorithm \ref{algo:bpgMIST}. Suppose Algorithm \ref{algo:bpgMIST} returns a MIST $T$. Recall that we visit the vertices in one of the orders $\alpha = (x_1, y_1, x_2, y_2, \ldots )$, or $\beta = (y_1, x_1, y_2, x_2, \ldots )$. Furthermore, any  vertex encountered during the execution of the algorithm must be pendant in $T$. As we never visit the same vertex twice, these pendant vertices are found in linear-time. The remaining graph must have a Hamiltonian path, and finding the Hamiltonian path is also linear-time in our algorithm. So, all the steps of Algorithm~\ref{algo:bpgMIST} can be executed in $O(n+m)$ time. Hence we have the following corollary.

\begin{corollary}
A maximum internal spanning tree of a bipartite permutation graph can be computed in linear-time.
\end{corollary}
%
\section{Bounds for Chain Graphs}
\label{sec:chain}
A bipartite graph $G = (X, Y,E)$ is a \textit{chain graph} if the neighborhoods of the vertices of $X$ form a chain, that is, the vertices of $X$ can be linearly ordered, say  $\{x_1,x_2, \ldots ,x_{n_1}\}$ such that $N(x_1) \subseteq N(x_2) \subseteq \ldots \subseteq N(x_{n_1})$ and $n_1 = \vert X\vert$. If $G=(X,Y,E)$ is a chain graph, then the neighborhoods of the vertices of $Y$ also form a chain. If $n_2 = \vert Y\vert$, an ordering $\alpha=(x_1,x_2, \ldots ,x_{n_1},y_1,y_2, \ldots ,y_{n_2})$ is called a \textit{chain ordering} if  $N(x_1) \subseteq N(x_2) \subseteq \ldots \subseteq N(x_{n_1})$ and $N(y_1) \supseteq N(y_2) \supseteq \ldots \supseteq N(y_{n_2})$. If a vertex $x_{i}$ appears before $x_{j}$ in chain ordering, we write $x_{i}<x_{j}$. Given a chain graph $G$, a chain ordering of $G$ can be computed in linear-time \cite{hegge}. Note that a chain ordering is also a strong ordering. So, every chain graph is also bipartite permutation graph.

 In this section, we will prove the following lower bound for number of internal vertices in a MIST of a chain graph G.
\begin{theorem}\label{thm:chainlowerbound}
For a chain graph $G$, let $P^*$ be an optimal path cover of $G$. Then $Opt(G) \geq \vert E(P^*)\vert - 2$.
\end{theorem}

In order to prove Theorem \ref{thm:chainlowerbound}, we look at  optimal path covers of bipartite permutation graphs. \cite{srik} gave an algorithm to find an optimal path cover of a bipartite permutation graph. Note that this algorithm applies to chain graphs as well. We will recall the algorithm given in \cite{srik}, but first we cover some notations used in the algorithm. A path cover $P^* = \{P_1, P_2, \ldots, P_k\}$ is \textit{contiguous} if it satisfies the following two conditions:
\begin{enumerate}
\item If $x\in X$ is the only vertex in $P_i$ and if $x' < x < x''$, then $x'$ and $x''$ belong to different paths.
\item If $xy$ is an edge in $P_i$ and $x'y'$ is an edge in $P_j$, where $i \neq j$ and $x < x'$, then $y < y'$.
\end{enumerate}

A path $P$ is \textit{contiguous} if it is one of the following forms:
\noindent
 $x_i y_j x_{i+1} y_{j+1} \ldots y_{t-1}x_r$, $x_i y_j x_{i+1} y_{j+1} \ldots y_{t-1}x_ry_t$,
 $y_jx_iy_{j+1}x_{i+1} \ldots x_{r-1}y_tx_r$, or $y_jx_iy_{j+1}x_{i+1} \ldots x_{r-1}y_t$ such that $r \ge i$ and $t \ge j$. Note that every path in a contiguous path cover is contiguous. Let $P$ be a contiguous path which ends with some edge, say $x_p y_q$. If $y_q x_{p+1} \notin E(G)$, then we say that the path $P$ is not extendable on the right. A contiguous path is said to be a \textit{maximal contiguous path} if it is not extendable on the right. An optimal path cover  $P^* = \{P_1, \ldots, P_k\}$ is a \textit{maximum optimal path cover} if each $P_i$ covers the maximum number of vertices in $V(G) \setminus \{P_1 \cup P_2 \cup \ldots P_{i-1}\}.$ According to \cite{srik}, there exists an optimal path cover which is a maximum optimal path cover for any bipartite permutation graph $G$ such that every path in the path cover is a maximal contiguous path.

As a chain graph is an instance of a bipartite permutation graph, we recall the algorithm from \cite{srik} which finds this desired maximum optimal path cover for a chain graph (Algorithm \ref{algo:optpathcover}). From this point, we will refer such a path cover as an optimal path cover only.

\begin{algorithm}[ht!]
 \caption{\label{algo:optpathcover} Algorithm for finding an optimal path cover of $G$}
 \textbf{Input:} {A chain graph $G=(X,Y,E)$ with the ordering of its vertices}\\
  \textbf{Output:} {An optimal path cover $P$ of $G$}
   \setcounter{AlgoLine}{0}

   \nl Mark all vertices in $X$ and $Y$ as not visited; let $P= \emptyset $.

   \nl \While{all vertices of $G$ are not visited}{

   \nl Let $x$ and $y$ be the first vertices in $X$ and $Y$ which are not visited.

   \nl Let $P_x$ and $P_y$ be the maximal contiguous paths starting from $x$ and $y$, respectively.

   \nl $Q$:= Maximum of $P_x$ and $P_y$.

   \nl $P:=P \cup Q$.

   \nl Mark all vertices in $Q$ as visited.}

   \nl Output $P$.
\end{algorithm}
Now, we give the proof of Theorem \ref{thm:chainlowerbound}.

\begin{proof}[Proof of Theorem~{\upshape\ref{thm:chainlowerbound}}]
Let $P^* $ be the optimal path cover obtained from Algorithm \ref{algo:optpathcover}.  A path in $P^*$ is nontrivial if it has at least two vertices.  We may assume that the path components of $P^*$ are ordered with respect to their appearance in Algorithm \ref{algo:optpathcover}. To complete the proof, we will construct a spanning tree by connecting the paths of $P^*$ with edges of $G$. Let $P$ and $Q$ be two consecutive nontrivial path components in $P^*$, and we will now find a sufficient edge connecting $P$ and $Q$. We have four cases to consider.\\
\\
\textbf{Case 1:} \underline{$P$ ends at $X$ side and $Q$ starts from $X$ side}\\
Suppose that $P$ ends at some vertex $x$ and $Q$ starts from some vertex $x'$, where $x < x'$. Let $y$ be the vertex adjacent to $x$ in $P$, then $y \in N(x')$ as $G$ is a chain graph. Here, we consider the edge $yx'$ as the combining edge for path components $P$ and $Q$. We see that $y$ is internal in $P$ and $x'$ is pendant in $Q$. Fig.~\ref{fig:casea} provides an illustration.
\begin{figure}[h!]
 \centering
 \includegraphics[width=6cm, height=2cm]{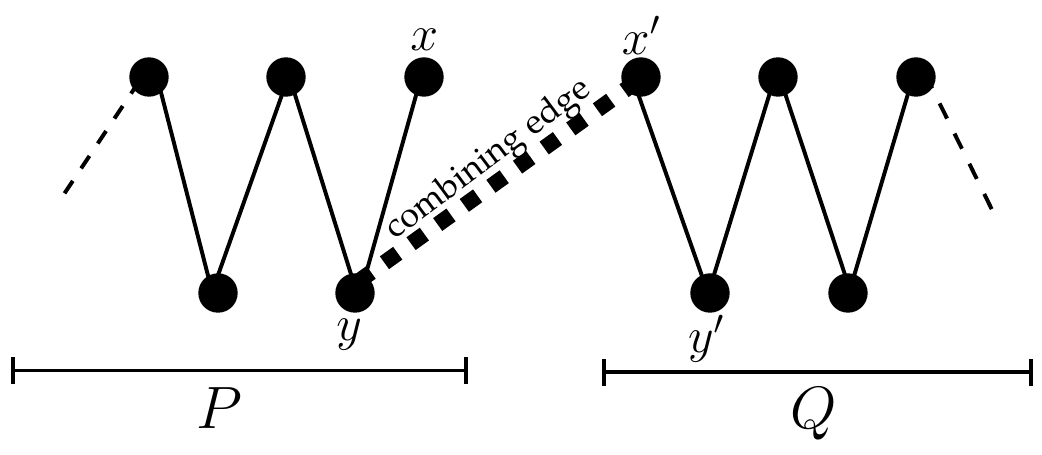}
 \caption{Case 1)}
\label{fig:casea}
\end{figure}

\noindent
\textbf{Case 2:} \underline{$P$ ends at $Y$ side and $Q$ starts from $Y$ side}\\
As $P^*$ was constructed from Algorithm \ref{algo:optpathcover}, every path component in $P^*$ is maximal contiguous. But, in this case, $P$ is extendable on right. So, this case will not arise. Fig.~\ref{fig:caseb} provides an illustration.\\
\begin{figure}[h!]
 \centering
 \includegraphics[width=6cm, height=1.5cm]{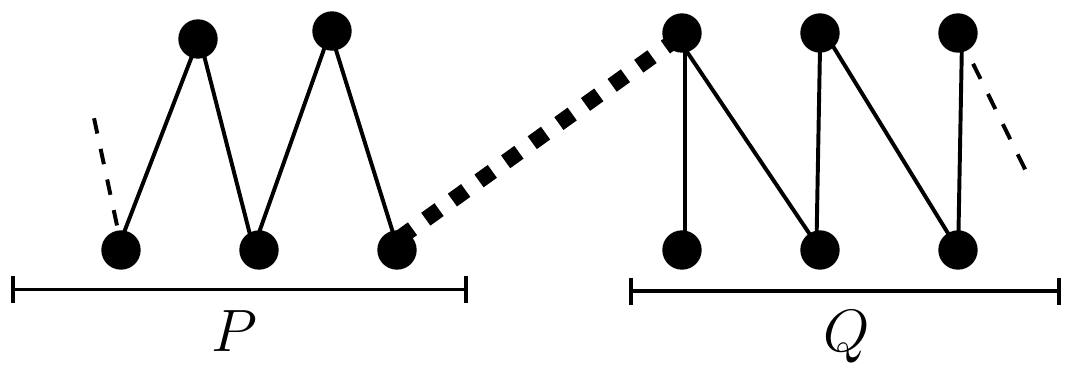}
 \caption{Case 2)}
\label{fig:caseb}
\end{figure}

\noindent
\textbf{Case 3:} \underline{$P$ ends at $Y$ side and $Q$ starts from $X$ side}\\
This case will also not arise. The reason is same as of Case 2. Fig.~\ref{fig:casec} provides an illustration.
\begin{figure}[h!]
\centering
 \includegraphics[width=6cm, height=1.5cm]{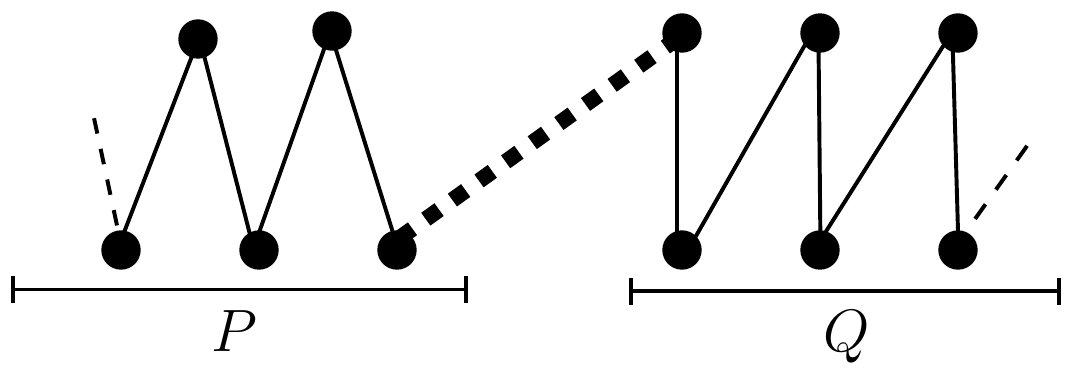}
 \caption{Case 3)}
\label{fig:casec}
\end{figure}

\noindent
\textbf{Case 4:} \underline{$P$ ends at $X$ side and $Q$ starts from $Y$ side}\\
Suppose that $P$ ends at some vertex $x$ and $Q$ starts from some vertex $y = y_j$. As the number of vertices in $G$ is finite, path $Q$ must end at $X$ or $Y$. Now, we consider two subcases depending on $Q$ ending at $X$ side or $Y$ side.

\noindent
\textbf{Subcase 4.1:} \underline{$Q$ ends at $X$ side}\\
 Let $Q = yx_iy_{j+1}\ldots x_ty_kx_{t+1}$. As $G$ is a chain graph, we have that $Q' = x_iyx_{i+1}\ldots y_{k-1}x_{t+1}y_k$ is also a path in $G$. Note that $V(Q) = V(Q')$ and $Q'$ is a maximal contiguous path. We can replace $Q$ with $Q'$ in the path cover $P^*$. Now as $Q'$ starts from $X$, we have reduced to Case 1). Fig.~\ref{fig:cased1} provides an illustration.

\begin{figure}[h!]
 \centering
 \includegraphics[width=7.5cm, height=4cm]{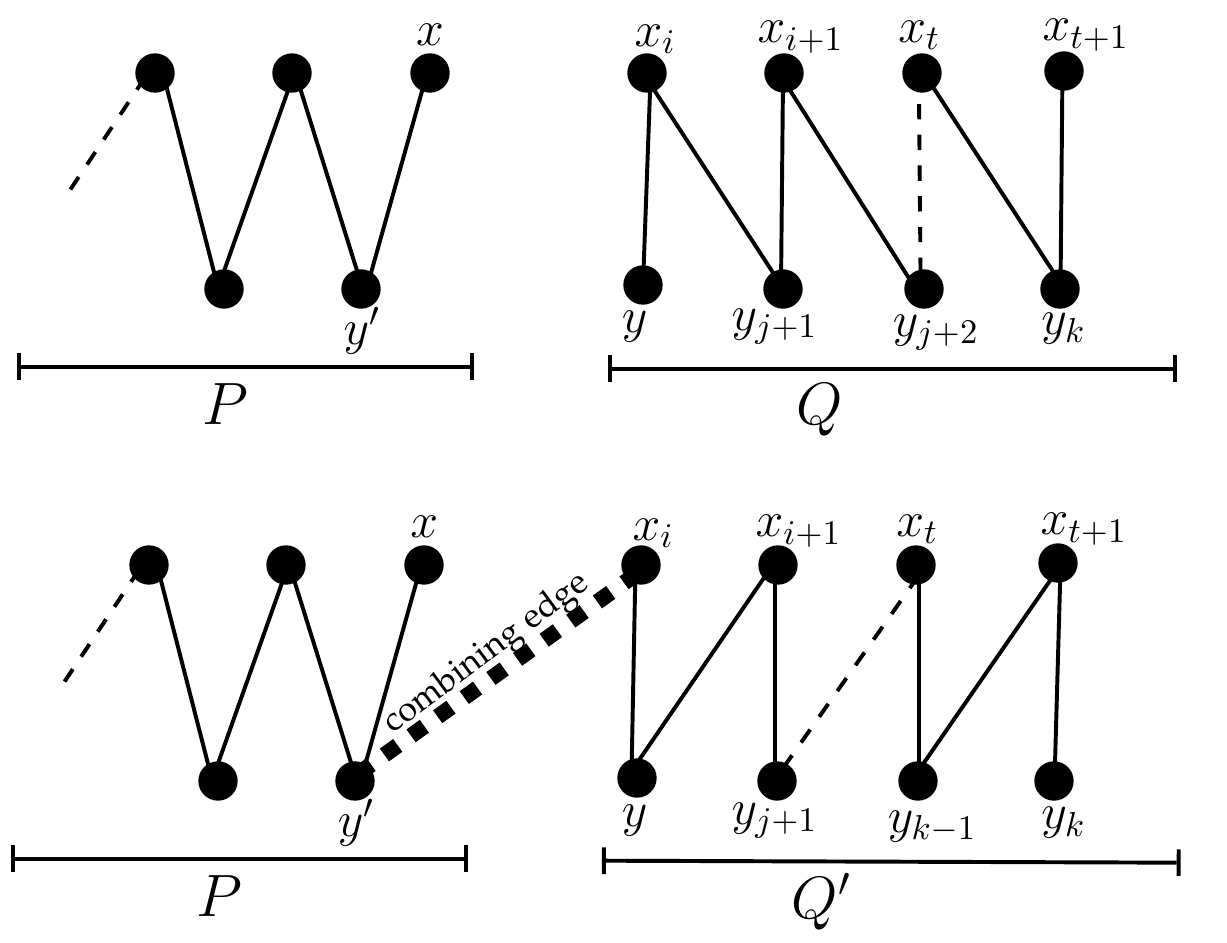}\\
 \caption{Subcase 4.1}
\label{fig:cased1}
\end{figure}

\noindent
\textbf{Subcase 4.2:} \underline{$Q$ ends at $Y$ side}\\
Let $y'$ be the neighbor of $x$ in $P$ and $x'$ be the neighbor of $y$ in $Q$. Since, $x < x'$ and $G$ is a chain graph, edge $y'x' \in E(G)$. Here, we consider the edge $y'x'$ as the combining edge for path components $P$ and $Q$. We see that $y'$ is internal in $P$ and $x'$ is also internal in $Q$. Fig.~\ref{fig:cased2} provides an illustration.

\begin{figure}[h!]
 \centering
 \includegraphics[width=7cm, height=2.5cm]{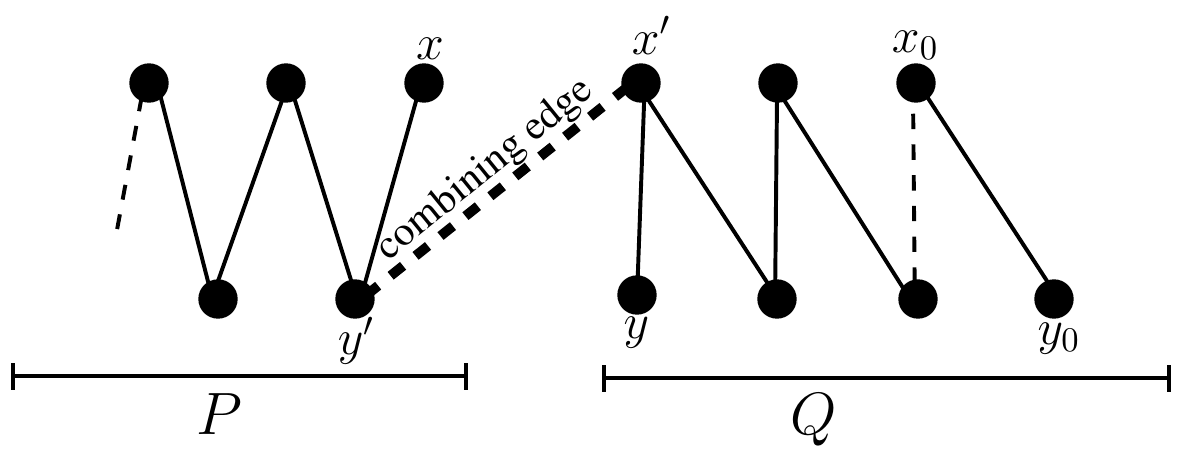}\\
 \caption{Subcase 4.2}
\label{fig:cased2}
\end{figure}

We see that in each possible case, we get a combining edge which connects both the path components $P$ and $Q$. If we connect each consecutive nontrivial path component with the combining edges and connect the remaining single vertex components by an arbitrary edge incident with an internal vertex of a nontrivial path component, we obtain a spanning tree of $G$.

Suppose $P^*$ has $k$ path components $P_1, P_2, \ldots, P_k$. Let us denote number of edges of the component $P_i$ by $e_i$ for every $1 \leq i \leq k$. This implies that $e_1 + e_2 + \ldots + e_k = \vert E(P^*)\vert$. Note that the number of internal vertices in a path with $e_{i}$ edges is $e_{i}-1$.

We now observe the case where $G$ is a graph such that Subcase 4.2 does not arise. Note then every combining edge connects one internal and one pendant vertex from different path components.  So, $i(T) = e_1 -1 + e_2 + e_3 + \ldots + e_k = \vert E(P^*)\vert-1$.
Now, suppose that Subcase 4.2 arises for some consecutive nontrivial paths $P$ and $Q$. Here, $Q$ ends at $Y$ side, say at $y_0$ and let $x_0$ be the neighbor of $y_0$ in $Q$. We claim that $x_0 = x_{n_1}$. If this is not the case then there exists a vertex $x^*$ in $X$ such that $x^* > x_0$ and $x^* \notin V(Q)$. But, since $G$ is a chain graph, we have that $(y_0 , x^*) \in E(G)$ which makes $Q$, a non-maximal path, a contradiction. Thus, $x_0 = x_{n_1}$  which implies that, if $Q{''} \in P^*$  and  appears after $Q$ in Algorithm \ref{algo:optpathcover}, then $Q''$ is a single vertex path component containing a vertex of $Y$. This implies that the Subcase 4.2 appears only once. So, $i(T) = e_1 -1 + e_2 + e_3 + \ldots + e_k - 1 = \vert E(P^*)\vert-2$. Hence, the number of internal vertices in any MIST of $G$ is at least $\vert E(P^*)\vert- 2$.
\end{proof}

Combining Theorem \ref{thm:general upper bound} and Theorem \ref{thm:chainlowerbound}, we can state the following corollary.
\begin{corollary}\label{corollary:chainopt(G)}
For a chain graph $G$, if $P^*$ denotes an optimal path cover then $Opt(G)$ is either $\vert E(P^*)\vert - 1$ or $\vert E(P^*)\vert - 2$.
\end{corollary}

Now, we give examples of chain graphs which shows that both the bounds (given by Theorem \ref{thm:general upper bound} and Theorem \ref{thm:chainlowerbound}) are tight. In Fig.~\ref{fig:chaintightbound}, $G_1$ and $G_2$ are chain graphs and $T_1$ and $T_2$ are Maximum Internal Spanning Trees of $G_1$ and $G_2$ respectively. We can see that optimal path cover obtained from Algorithm \ref{algo:optpathcover} for the graph $G_1$ is $\{x_1y_1x_2y_2x_3, y_3x_4y_4x_5y_5\}$ which has $8$ edges and its MIST $T_1$ has $6$ internal vertices i.e. $Opt(G_1) = \vert E(P^*)\vert- 2 = 8 - 2 = 6$. Using Observation \ref{Ob:basic}, it can be verified that any MIST of $G_1$ has at least four pendant vertices, two from $X$ side and two from $Y$ side; so, $G_1$ can have at most $6$ internal vertices in its MIST. Hence, $T_1$ is indeed a MIST of $G_1$. In a similar manner, optimal path cover obtained from Algorithm \ref{algo:optpathcover} for the graph $G_2$ is $\{x_1y_1x_2y_2x_3, y_3x_4y_4x_5y_5x_6\}$ which has $9$ edges and its MIST $T_2$ has $8$ internal vertices i.e. $Opt(G_2) = \vert E(P^*)\vert- 1 = 9 - 1 = 8$. \\
\begin{figure}
\centering
 \includegraphics[width=9cm, height=4cm]{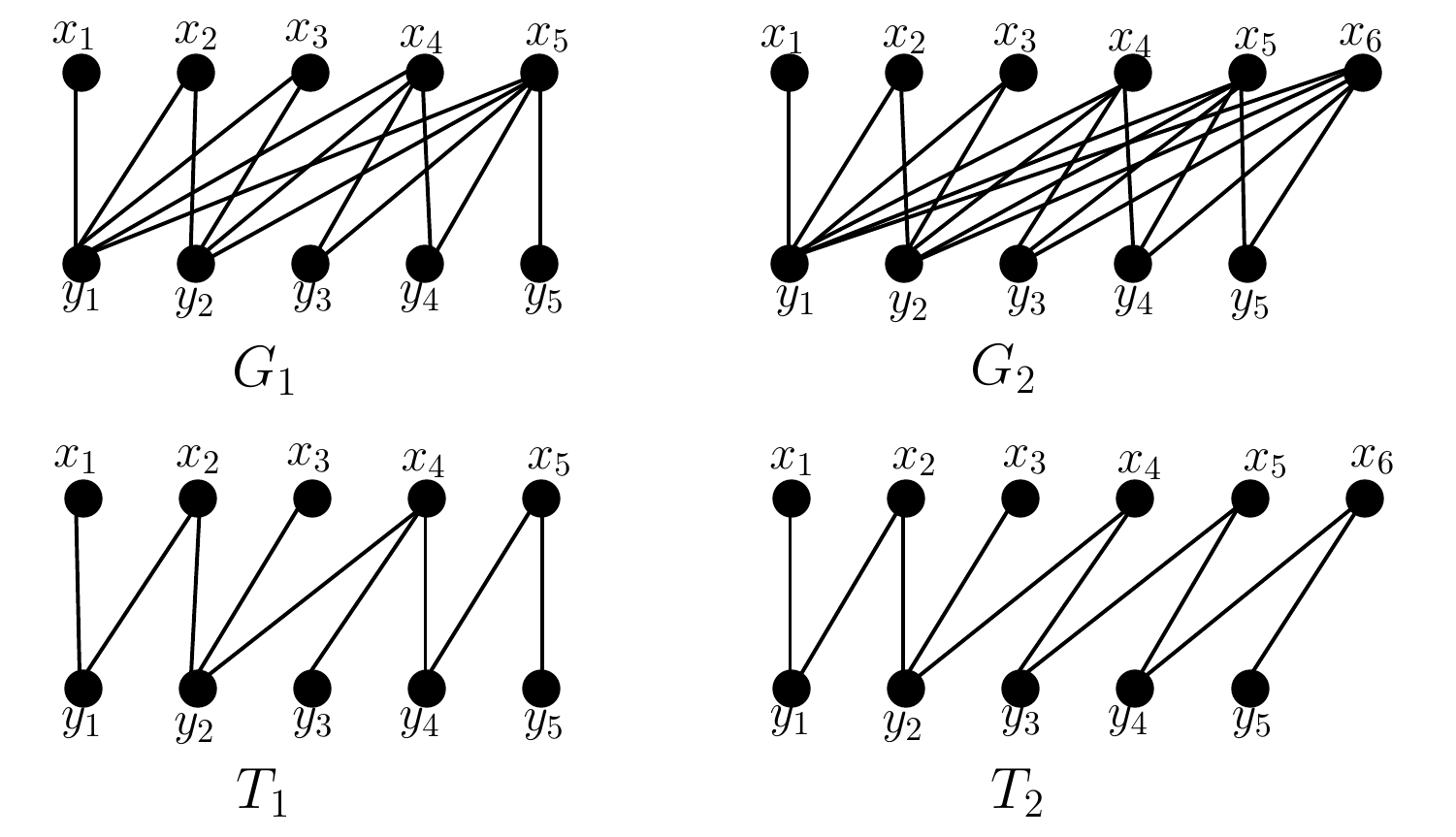}
 \caption{examples showing that bounds are tight}
\label{fig:chaintightbound}
\end{figure}

Corollary \ref{corollary:chainopt(G)} states that $\vert E(P^*)\vert - 2 \le Opt(G) \le \vert E(P^*)\vert - 1$ where $P^*$ is an optimal path cover of a chain graph $G$. We now argue that if $G$ is a bipartite permutation graph, then $Opt(G)$ cannot be lower bounded with value $\vert E(P^*)\vert - k$ for any fixed natural number $k$.  Below, for every natural number $k$, we give a construction of a bipartite permutation graph such that $Opt(G) = \vert E(P^*)\vert - O(5k)$.

For every integer $n = 5k ~ ( k \geq 1)$, we construct a connected bipartite permutation graph $G_n$  with $n$ vertices and $Opt(G_n) = \vert E(P^*)\vert  - O(n) $.  For all $1 \le i \le k$, let $X_i=\{x_1^i, x_2^i\}$ and  $Y_i = \{y_1^i, y_2^i, y_3^i\}$ if $i$ is even and $X_i=\{x_1^i, x_2^i, x_3^i\}$ and $Y_i = \{y_1^i, y_2^i \}$ for odd $i$. Let $V(G_n) = V_1 \cup V_2 \cup \ldots \cup V_k$ where $V_i = X_i \cup Y_i$ for all $1 \le i \le k$. Let $E(G_n) = E_1 \cup E_2 \cup \ldots  \cup E_k \cup E'$ where $E_i = \{xy \vert x \in X_i, y \in Y_i\}$ for each $1 \le i \le k$ and $E'$ is the set of edgs of the form $y_2^ix_1^{i+1}$ if $i$ is odd and $x_2^iy_1^{i +1}$ if $i$ is even for each $1 \le i \le (k-1)$. We see that $G_n$ is a bipartite permutation graph with $n$ vertices and $n+2k-1$ edges. Algorithm \ref{algo:optpathcover} gives an optimal path cover $P^*$ for $G_n$ having $4k$ edges and Algorithm \ref{algo:bpgMIST} gives a MIST with $3k$ internal vertices. So, we get that $Opt(G_n) = 3k = 4k - k = 4k - \frac{n}{5} = \vert E(P^*)\vert - O(n)$. Fig.~\ref{fig:bpbound} provides an illustration for $G_{25}$.
\begin{figure}
 \centering
 \includegraphics[page =2,width=12cm, height=7cm]{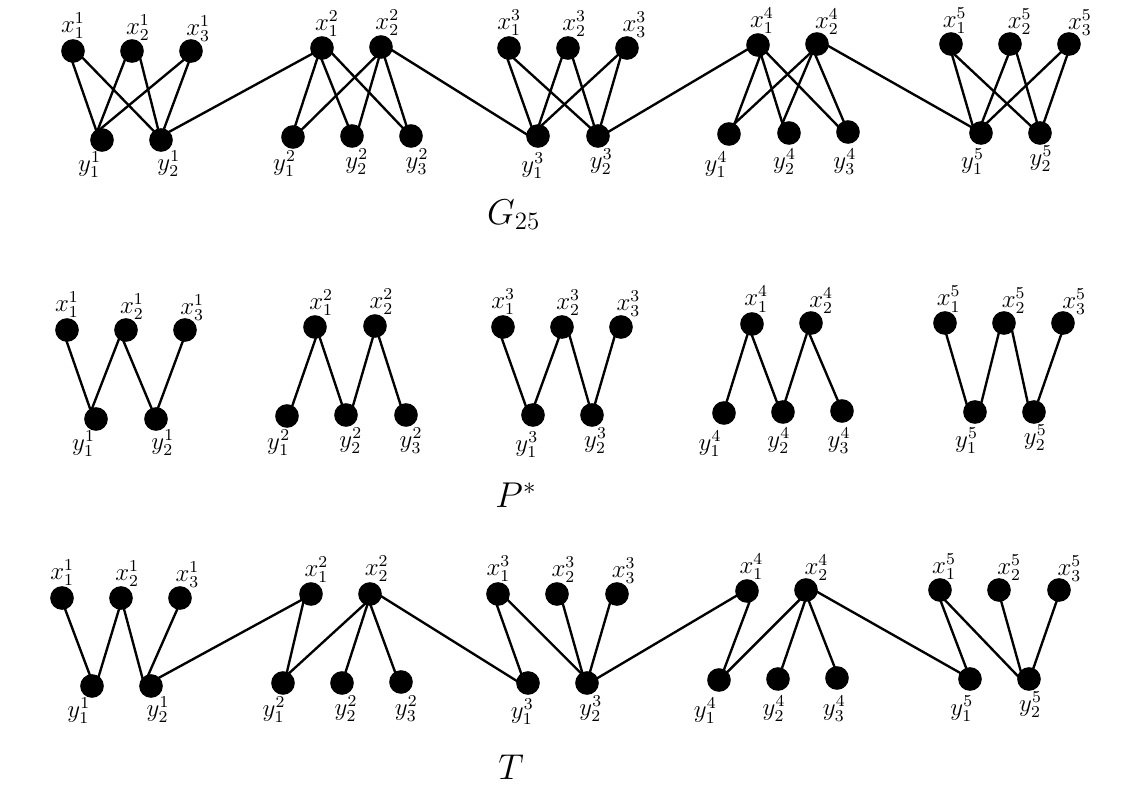}\\
 \caption{Graph $G_{25}$, its optimal path cover $P^*$ from Algorithm \ref{algo:optpathcover} and its MIST $T$ from Algorithm \ref{algo:bpgMIST}}
\label{fig:bpbound}
\end{figure}

Thus $Opt(G)$ for bipartite permutation graphs do not have lower bound of the form $\vert E(P^*)\vert - k$ for some fixed natural number $k$, independent of $n$.

\section{Conclusion}
\label{sec:conclusion}
We studied the Maximum Internal Spanning Tree (MIST) problem, a generalization of Hamiltonian path problem. As the MIST problem remains NP-hard even for bipartite graphs and chordal graphs due to a reduction from the Hamiltonian path problem \cite{laii, mull}, we further investigated the complexity of special instances of these classes, chain graphs, bipartite permutation graphs and block graphs. We also investigated cactus graphs and cographs, finding linear-time algorithms for the MIST problem for each of these graph classes.

\cite{li2018} proved an upper bound for $Opt(G)$ in terms of an optimal path cover.  We further studied this relationship between path covers and $Opt(G)$ and showed tight lower bounds for chain graphs and cographs. We also showed this phenomenon does not hold for general graphs with a construction of bipartite permutation graph and block graph such that $Opt(G)$ is arbitrarily far from $\vert E(P^*)\vert$.

A convex bipartite graph $G$ with bipartition $(X,Y)$ and an ordering $X=(x_1,x_2, \ldots, x_n)$, is a bipartite graph such that for each $y \in Y$, the neighborhood of $y$ in $X$ appears consecutively. Complexity status of the MIST problem is still open for convex bipartite graphs, which is a superclass of bipartite permutation graphs and subclass of chordal bipartite graphs.  Designing an algorithm for MIST in convex bipartite graphs will be a good research direction.

The weighted version of the MIST problem is also well studied in literature \cite{sala}. Given a vertex-weighted connected graph $G$, the maximum weight internal spanning tree (MwIST) problem asks for a spanning tree $T$ of $G$ such that the total weight of internal vertices in $T$ is maximized. Since MwIST problem is a generalization of the MIST problem, one may also investigate the complexity status of MwIST problem for some special classes of graphs.

To our knowledge, every known hardness proof for the MIST problem on  families of graphs relies on a reduction to Hamiltonian path problem. We leave as an open question if there exists a family of graphs such that Hamiltonian path problem is polynomial time, yet the MIST problem remains NP-hard.

%
%

\end{document}